\documentclass[11pt,a4paper,cleveref,autoref,thm-restate]{article}
\usepackage[margin=0.96in]{geometry}
\usepackage{graphicx,framed,bbm,tabularx,amssymb,amsmath,amsthm,color,tikz}
\usepackage{pdflscape,latexsym,bm,array,caption,textcomp,url,cite,setspace}
\usepackage[boxruled,vlined]{algorithm2e}
\usepackage{graphicx,subfigure}

\usepackage[hang,flushmargin]{footmisc}
\usepackage{hyperref}
\hypersetup{colorlinks,breaklinks,urlcolor=[rgb]{0,0,0},linkcolor=[rgb]{0,0,0.8},citecolor=[rgb]{0,0,0.8}}
\usepackage[capitalise]{cleveref}
\usepackage[capitalise]{cleveref}
\usepackage{framed}

\newenvironment{reminder}[1]{\smallskip
\noindent {\bf Reminder of #1  }\em}{}

\newcolumntype{C}{>{\centering\arraybackslash}p{3.9em}}
\newcommand{\zell}[2][c]{ \begin{tabular}[#1]{@{}c@{}}#2\end{tabular}}

\title{
Maximum-Flow and Minimum-Cut Sensitivity Oracles for \\ Directed Graphs
\vspace{2mm}
}

\author{
Mridul Ahi \thanks{Indian Institute of Technology, Delhi, India; {\tt mt1210901@maths.iitd.ac.in}}
\and 
Keerti Choudhary \thanks{Indian Institute of Technology, Delhi, India; {\tt keerti@iitd.ac.in}}
\and 
Shlok Pande \thanks{Indian Institute of Technology, Delhi, India; {\tt cs1210563@cse.iitd.ac.in}}
\and 
Pushpraj \thanks{Indian Institute of Technology, Delhi, India; {\tt cs5210596@cse.iitd.ac.in}}
\and 
Lakshay Saggi \thanks{Indian Institute of Technology, Delhi, India; {\tt csz228231@cse.iitd.ac.in}}
}

\date{}

\newcommand{\A}{{\cal A}}
\newcommand{\B}{{\cal B}}
\newcommand{\C}{{\cal C}}

\newcommand{\E}{{\cal E}}
\newcommand{\G}{{\cal G}}

\renewcommand{\P}{{\cal P}}

\newcommand{\R}{{\cal R}}

\newcommand{\W}{{\cal W}}
\newcommand{\x}{\bm{x}}

\newcommand{\inv}{\textsc{IN}}
\newcommand{\outv}{\textsc{OUT}}

\newcommand{\inedges}{\textsc{In-Edges}}
\newcommand{\outedges}{\textsc{Out-Edges}}
\newcommand{\maxflow}{\textsc{max-flow}}
\newcommand{\mincut}{\textsc{min-cut}}
\newcommand{\nmc}{\textsc{NMC}}
\newcommand{\fmc}{\textsc{FMC}}

\renewcommand{\geq}{\geqslant}
\renewcommand{\leq}{\leqslant}

\newcommand{\Null}{\textsc{null}}
\newcommand{\rev}{\mathrm{rev}}

\newtheorem{theorem}{Theorem}[section]
\newtheorem*{theorem*}{Theorem}
\newtheorem{lemma}[theorem]{Lemma}

\newtheorem{observation}[theorem]{Observation}

\newtheorem{corollary}[theorem]{Corollary}
\newtheorem{property}[theorem]{Property}

\newtheorem{assumption}[theorem]{Assumption}

\newtheorem{definition}[theorem]{Definition}

\newtheorem{question}{Question}

\newtheorem*{question*}{Question}
\newtheorem*{problem*}{Problem}

\begin{document}
\sloppy

\maketitle

\thispagestyle{empty}

\vspace{1mm}

\begin{abstract}
\noindent
This paper addresses the problem of designing fault-tolerant data structures for the $(s,t)$-max-flow and $(s,t)$-min-cut problems in unweighted directed graphs. Given a directed graph $G = (V, E)$ with a designated source $s$, sink $t$, and an $(s,t)$-max-flow of value $\lambda$, we present novel constructions for  max-flow and min-cut sensitivity oracles, and introduce the concept of a fault-tolerant flow family, which may be of independent interest. Our main contributions are as follows.

\medskip

\noindent
{\em 1. Fault-Tolerant Flow Family:}
For any graph $G$ with $(s,t)$\nobreakdash-max-flow value $\lambda$, we construct a family $\mathcal{B}$ of $2\lambda+1$ $(s,t)$-flows such that for every edge $e$, $\mathcal{B}$ contains an $(s,t)$\nobreakdash-max-flow of $G-e$. This covering property is tight up to constants for single failures and provably cannot extend to comparably small families for $k \ge 2$, where we show an $\Omega(n)$ lower bound on the  family size, independent of $\lambda$. 

\medskip

\noindent
{\em 2. Max-Flow Sensitivity Oracle:}
Using the fault-tolerant flow family, we construct a single as well as dual-edge sensitivity oracle for $(s,t)$-max-flow that requires only $O(\lambda n)$ space. Given any set $F$ of up to two failing edges, the oracle reports the updated max-flow value in $G-F$ in $O(n)$ time. Additionally, for the single-failure case, the oracle can determine in constant time whether the flow through an edge $x$ changes when another edge $e$ fails. To the best of our knowledge, these are the first non-trivial constructions of fault-tolerant max-flow oracle.

\medskip

\noindent
{\em 3. Min-Cut Sensitivity Oracle for Dual Failures:}
Recently, Baswana et al. (ICALP’22) designed an $O(n^2)$-sized oracle for answering $(s,t)$-min-cut size queries under dual edge failures in constant time, along with a matching lower bound. We extend this by focusing on graphs with small min-cut values $\lambda$, and present a more compact oracle of size $O(\lambda n)$ that answers such min-cut size queries in constant time and reports the corresponding $(s,t)$-min-cut partition in $O(n)$ time. We also show that the space complexity of our oracle is asymptotically optimal in this setting.

\medskip

\noindent
{\em 4. Min-Cut Sensitivity Oracle for Multiple  Failures:}
We extend our results to the general case of $k$ edge failures. For any graph with $(s,t)$-min-cut of size $\lambda$, we construct a $k$-fault-tolerant min-cut oracle with space complexity $O_{\lambda,k}(n \log n)$ that answers min-cut size queries in $O_{\lambda,k}(\log n)$ time. This also leads to improved fault-tolerant $(s,t)$-reachability oracles, achieving $O(n \log n)$ space and $O(\log n)$ query time for up to $k = O(1)$ edge failures. This is the first such construction to support more than two failures with near-linear space and sublinear query time, improving significantly upon previous work.

\end{abstract}

\newpage

\setcounter{page}{1} 

\section{Introduction}

The increasing scale and complexity of modern networks necessitate the development of efficient algorithms for solving fundamental graph problems. 
Among these problems, the {\em max-flow} and {\em min-cut} problems have received significant attention due to their broad applicability in transportation networks, data propagation, resource allocation, and numerous other domains. 
Since the seminal work of Ford and Fulkerson~\cite{FF62}, extensive research has focused on computing max-flow in networks~\cite{EdmondsK72, GoldbergT88, GoldbergR98, Dinitz06, LeeRS13, KathuriaLS20, BrandLYLSSSW21, GaoLP22, BrandGJLLPS22, BrandCKLPGSS23, BernsteinBST24}, culminating in the near-linear-time algorithm by Chen et al.~\cite{ChenKLPGS23}.
While the problem of computing max-flow and min-cut in a graph has been extensively studied in the past 70 years for static graphs, little is known about designing efficient data-structures for these problems for networks prone to failures.

In recent years, there has been significant interest in fault-tolerant data structures, motivated by the fact that real-world networks are often vulnerable to node or link failures.
In the fault-tolerant model, we assume that the total number of failures at any point remains bounded by a parameter~$k$, which depends on the robustness of the network. Typically, $k$ is much smaller than the total number of vertices in the network. Over the past two decades, many efficient fault tolerant data-structures have been developed for various classical graph problems, 
including spanners~\cite{Luk99, DK11, Par14}, distance preservers~\cite{ParterP:13, BodwinGPW:17, GK17}, distance oracles~\cite{DT:08,CLPR10,DP09,BernsteinK:08}, reachability ~\cite{BaswanaCR:16,Choudhary16,ChakrabortyC20}, connectivity~\cite{PatrascuT:07, DuanP:10, DuanP:17,GIP17,BCR19}, etc.

Despite these advances, the domain of max-flow and min-cut problems has seen limited progress in the context of network failures.
This paper aims to address this gap by studying the problem of flows and cuts in directed networks susceptible to edge failures. We introduce simple yet powerful techniques for the aforementioned problems, which lead to state-of-the-art results for the specific problems which we study.    
Our objective is to compute a compact data structure, referred to as a {\em sensitivity} or {\em fault-tolerant} oracle, which, after the failure of any set of edges in the input graph $G$ efficiently determines the max-flow and min-cut with respect to a source $s$ and sink $t$. 
This problem is formalized as follows.

\begin{question}
Given a directed, unweighted graph $G$ with a designated source $s$ and a designated sink $t$, design a data-structure that given any query set of edges $F \subseteq E$ of size $k$, outputs the $(s,t)$-max-flow and $(s,t)$-min-cut on deletion of edges in $F$ from graph $G$.
\end{question}

\subsection{Our Contributions}

\paragraph{Max-flow Sensitivity Oracle for Single and Dual failures}
The first result we discuss pertains to the max-flow problem in directed graphs.
To date, no efficient sensitivity oracles for max-flow are known, even for handling a single edge failure. While it is possible to determine efficiently whether the $(s,t)$-max-flow value decreases after an edge failure—thanks to the work of Picard and Queyranne~\cite{PicardQ82}, which yields an $O(n)$-sized data structure with $O(1)$ query time (see Theorem 4.6 of~\cite{BaswanaBP22})—the problem of efficiently reporting the updated $(s,t)$-max-flow after an edge failure remains unresolved. The complexity of this problem arises from the potential need to reroute flow through an alternative path in the residual graph, especially if the failed edge was carrying flow in the original $(s,t)$-max-flow of $G$. 

To tackle this, we first demonstrate that for any directed graph $G$, there exists a small family of flows that effectively captures the maximum flows of the graph $G - e$ for all potential edge failures $e \in E$. Specifically, we obtain the following result.

\begin{theorem}
For any directed graph $G=(V,E)$ with an $(s,t)$-max-flow of value $\lambda$, there exists a family $\mathcal{B}$ of $2\lambda+1$ $(s,t)$-flows in $G$ satisfying the following property:
\begin{center}
For each edge $e$, the family $\mathcal{B}$ includes an $(s,t)$-max-flow for the graph $G - e$.
\end{center}
\label{thm-intro:FT-flow}
\end{theorem}


We further prove an $\Omega(n)$ lower bound on the cardinality of $\B$, independent of $\lambda$, for $k \geq 2$ failures.

It is important to note that the above theorem provides a compact flow family resilient to a single edge failure, which can be amenable to various algorithmic applications. However, it does not directly imply a compact data structure for efficiently reporting the maximum flow in a fault-prone network. 
We thus ask the following question.

\begin{question}
Is it possible to design a compact data structure that, upon the failure of any set $F \subseteq E$, possibly of size one, can efficiently update the $(s,t)$-max-flow of the graph?
\end{question}

To the best of our knowledge, this problem has not been addressed in the existing literature. We answer this question affirmatively by presenting the following sensitivity oracle for handling both single and dual edge failures, which leverages the flow-covering result from \Cref{thm-intro:FT-flow}.

\begin{theorem}
For any directed graph $G = (V, E)$ with $n$ vertices, there exists an $(s,t)$-max-flow $f$ and a data structure of size $O(\lambda n)$, where $\lambda$ is the flow value, such that upon the failure of any set $F \subseteq E$ of up to two edges, the data structure can implicitly report the updated $(s,t)$-max-flow $f'$ for $G - F$ in $O(n)$ time, that is, it can identify all edges $e$ in $G - F$ for which $f'(e) \neq f(e)$ in $O(n)$ time. 

Moreover, for the case of single edge failure $(|F| = 1)$, the data structure can determine for any edge $e$ in $G - F$ whether $f'(e)$ is $0$ or $1$ in $O(1)$ time.
\label{theorem:FT-flow-4}
\end{theorem}

We note that the bound in \Cref{theorem:FT-flow-4} is essentially tight, as $\Omega(\lambda n)$ is also a lower bound on the number of edges required to represent a maximum flow from source to sink. 

\smallskip

\paragraph{Min-Cut Sensitivity Oracle for Dual failures}
We next discuss our result on min-cut sensitivity oracle for two edge failures. Baswana~et~al.~\cite{BaswanaBP22} presented an $O(n^2)$-sized data structure to answer min-cut queries after two edge failures in $O(1)$ time. They also gave a conditional lower bound of $\widetilde{\Omega}(n^2)$ on the oracle size, based on the {\em Directed Reachability Hypothesis}. Recently, Bhanja~\cite{Bhanja24} improved this lower bound by proving an unconditional bound of $\Omega (n^2)$ on the size of the oracle. 
This lower bound holds for directed as well as undirected graphs.

Given that the result of \cite{BaswanaBP22} has a matching upper and lower bound of $\Theta(n^2)$, any improvement to the dual edge failure oracle seems unlikely. However, in many applications, the focus is towards networks with cuts of small size. Similar scenarios have been previously studied by Abboud~et~al.~\cite{AbboudGIKPTUW19}, Akmal and Jin~\cite{AkmalJ23} for finding all pairs of vertices in graphs with small edge connectivity. 
This naturally raises the following question.%
\footnote{We remark that the data-structure of \cite{BaswanaBP22} takes quadratic size even when $\lambda=1$. For details, see \Cref{section:baswana-min-cut-oracle-space}.}

\begin{question}
Given a directed graph $G$ with an $(s,t)$-min-cut of size $O(n^{1-\epsilon})$ for some $\epsilon>0$,
can we construct a sub-quadratic sized oracle for efficiently reporting the $(s,t)$-min-cut size after dual edge failures?
\end{question}

We answer this question in the affirmative by presenting the following result, where the size of the oracle scales linearly with the min-cut size $\lambda$.

\begin{theorem}
For any directed graph $G$ with $n$ vertices and an $(s,t)$-min-cut of size $\lambda$, there exists a dual fault-tolerant min-cut oracle of size $ O(\lambda n) $ that, given any set $F$ of two edge failures, reports the size of the $(s,t)$-min-cut in the graph $G - F$ in $O(1)$ time.
Furthermore, the oracle reports an $(s,t)$-min-cut partition in $G - F$ in $O(n)$ time.
\label{theorem:dual-min-cut}
\end{theorem}

Our min-cut oracle for dual failures employs a different set of tools compared to those in \cite{BaswanaBP22}. In particular, we show the robustness of our flow-covering result of \Cref{thm-intro:FT-flow} by combining it with the seemingly unrelated 1-fault-tolerant strong-connectivity oracle of Georgiadis et al.~\cite{GeorgiadisIP20}, that takes linear space and answers strong-connectivity queries after single edge failure in constant time. This combination not only allows us to design an alternate min-cut sensitivity oracle but also introduces a new set of ideas which are of independent interest.

We further show that the bound of $\lambda n$ on the size of the oracle in \Cref{theorem:dual-min-cut} is tight by presenting the following result.

\begin{theorem}
For any positive integers $n, \lambda$ with $n \geq \lambda$, there exists a directed graph $G$ on $n$ vertices, with source $s$, sink $t$, and an $(s,t)$-min-cut of size $\lambda$, such that any dual fault-tolerant min-cut oracle for $G$ requires $\Omega(\lambda n)$ space. 
\label{theorem:lb-min-cut}
\end{theorem}

\paragraph{Fault-Tolerant Min-Cut Oracle for general $\bm{k}$ failures} 
We now discuss the problem of designing fault-tolerant min-cut oracle for graphs under the general setting of up to $k$ failures. Our goal is to construct an oracle of $F(k, \lambda) \cdot o(n^2)$ space and $F'(k, \lambda)$ query time. Before presenting our results for this scenario, it is essential to set the context by discussing a broader landscape. 

The task of building a $k$-fault-tolerant $(s,t)$-min-cut oracle is notably more complex than designing a $k$-fault-tolerant $(s,t)$-reachability oracle%
\footnote{This is because $(s,t)$-reachability can be reduced to $(s,t)$-min-cut by adding a dummy source with an out-edge to the original source, limiting the maximum flow to one.}.
At present, optimal-size $(s,t)$-reachability oracles in the fault-tolerant setting have been achieved only for single and dual failures. See \Cref{table:reachability-results-overview} for a brief summary of existing results.
For $k = 1$, the dominator trees introduced by Lengauer and Tarjan~\cite{LengauerT:79} provides a linear-size data structure with constant query time.
Choudhary~\cite{Choudhary16} extended this by presenting an optimal dual fault-tolerant single source reachability oracle.
For general $k$, the situation becomes significantly more challenging. The only currently known results are an oracle with $O(2^k n)$ query time due to fault-tolerant reachability preservers of \cite{BaswanaCR:16}; and an all-pairs reachability oracle by Brand and Saranurak~\cite{BrandS19}, which achieves a size of $O(n^2 \log n)$ and has $O(k^\omega)$ query time.

\begin{table}[!ht]
\centering
\resizebox{0.96\textwidth}{!}{ 
\setstretch{1.3}
\begin{tabular}{|c|c|c|c|c|}
\hline
\bf{Problem} & \bf{Failures} & \bf{Size of the oracle} & \bf{Query time} & \bf{Reference} \\
\hline
$s$-Reachability & 1 &  $O(n)$ & $O(1)$ & Lengaur \& Tarjan~\cite{LengauerT:79} \\
$s$-Reachability & 2 & $O(n)$ & $O(1)$ & Choudhary~\cite{Choudhary16} \\
All-Pairs Reachability & $k$ &  $O(n^2 \log n)$ & $O(k^\omega)$ & Brand \& Saranurak \cite{BrandS19} \\
$s$-Reachability & $k$ &  $O(2^k n)$ & $O(2^k n)$ & Baswana et al. \cite{BaswanaCR:16} \\
\hline
\end{tabular}
}
\caption{A summary of existing fault-tolerant reachability oracles.} 
\label{table:reachability-results-overview}
\end{table}


This naturally leads us to the following open question.

\begin{question}
Does there exist a linear-size oracle with constant query time for the $(s,t)$-reachability problem under $k \geq 3$ failures?
\end{question}

In this work, we affirmatively address this question by providing near-optimal results for $(s,t)$-reachability under any constant number of failures. In fact, we solve the more general problem of the $k$-fault-tolerant min-cut, as formalized in the following theorem.

\begin{theorem}
For any $n$ vertex directed graph $G$ with an $(s,t)$-min-cut of size $\lambda$, there exists a $k$-fault-tolerant minimum cut oracle requiring  $O(2^{O(L^4 \log L)} n \log n)$  
space, where $L = \lambda + k$. This oracle can compute the size of the $(s,t)$-min-cut in the graph $G - F$, for any set $F$ of~$k$ edges, in  $O(2^{O(L^4 \log L)}\log n)$  
time. Furthermore, the oracle can report an $(s,t)$-min-cut in $G - F$ in an additional computation time of $O(kn)$.
\label{theorem:k-min-cut}
\end{theorem}

Although the space complexity and query time of our oracle grow significantly with respect to $ k $, the oracle becomes highly efficient for small values of $ k $. Similar assumptions have been made in prior work. For instance, recent advancements \cite{WeimannY13, BrandS19, DuanR22, KarczmarzS23} in fault-tolerant distance oracles for undirected graphs have resulted in a breakthrough by Dey and Gupta~\cite{DeyG24} where they design an oracle with $ O(k^4 n^2) $ space and query time of $(k \log n)^{O(k^2)}$.\\[-2mm]

A summary of existing results on $(s,t)$-min-cut sensitivity oracles, along with our new contributions, is presented in \Cref{table:min-cut-results-overview}.

\begin{table}[!h]
\centering
\resizebox{0.96\textwidth}{!}{ 
\setstretch{1.3}
\begin{tabular}{|c|c|c|c|c|}
\hline
\textbf{Failures} & \textbf{Size of the oracle} & \textbf{\zell{Query time}} & \textbf{Reference} \\ \hline
$1$ & $O(n)$ & $O(1)$  & Picard et al. \cite{PicardQ82} \\ 
$2$ & $O(n^2)$ & $O(1)$ & Baswana et al. \cite{BaswanaBP22} \\ 
$2$ & $O(\lambda n)$ & $O(1)$ & \Cref{theorem:dual-min-cut} \\ 
general $k$ & $~O(2^{O((\lambda+k)^4 \log {(\lambda+k)})} n\log n)~$ & $ ~O(2^{O((\lambda+k)^4 \log {(\lambda+k)})}\log n)~$ & \Cref{theorem:k-min-cut} \\ \hline
\end{tabular}
}
\caption{A summary of $(s,t)$-min-cut sensitivity oracles for directed graphs.}
\label{table:min-cut-results-overview}
\end{table}

As a direct consequence of \Cref{theorem:k-min-cut}, we obtain the following corollary.

\begin{corollary}  
For any directed graph $G=(V,E)$ with $n$ vertices, a source $s$, a sink $t$, and any constant integer $k \geq 1$, there exists an $(s, t)$-reachability oracle that takes $O(n \log n)$ space and determines in $ O(\log n)$ time whether a path exists from $s$ to $t$ in the graph $G - F$ for any query set $F \subseteq E$ of at most $k$ edges.  
\label{corollary:reachability-k-failures}  
\end{corollary}

\subsection{Other Related Works}
This subsection provides an overview of additional related works.
For directed weighted graphs, Baswana and Bhanja~\cite{BaswanaB24} proposed a novel min-cut oracle that efficiently handles single edge failure. Their oracle takes $O(n)$ space and can return the size of the $(s,t)$-min-cut after failure of any edge in constant time. Additionally, they provide an $O(n^2)$ space data structure for finding the resulting cut in $O(n)$ time. This result, as well as the earlier work of Baswana and Bhanja~\cite{BaswanaB24} for dual failures in the unweighted setting, holds for both directed as well as undirected graphs.

We now turn to discuss the existing results specifically for the setting of undirected graphs. Baswana and Pandey~\cite{BaswanaP22} designed the first all-pairs single-edge sensitivity oracle for undirected unweighted graphs that takes $O(n^2)$ space and answers any query in O(1) time. Currently, the problem of extending this result to dual failures or weighted graphs remains open. 

Additionally, single fault-tolerant oracles for Steiner and global min-cuts exist as byproducts of the cactus representation~\cite{DinicKL76}, which stores all global min-cuts, and the connectivity carcass~\cite{DinitzV94, DinitzV95}, which stores Steiner min-cuts, as noted in~\cite{Bhanja24, BaswanaP22}. Recently, Bhanja~\cite{Bhanja24} gave efficient sensitivity oracles for Steiner and global min-cuts under single edge failure in undirected weighted graphs.

A closely related problem to fault-tolerant min-cut is dynamic min-cut, where the goal is to maintain min-cut information as the graph undergoes edge insertions or deletions. For global min-cuts, early work such as Thorup's tree-packing technique~\cite{Thorup07} laid the foundation, while recent advances using expander-decomposition-based methods~\cite{GoranciHNSTW23, JinST24} have significantly improved the query time. For all-pairs min-cuts, Jin and Sun~\cite{JinS22} achieve state-of-the-art results with a subpolynomial update time of $n^{o(1)}$ for graphs with small min-cut size. Further recent developments for flows in incremental and decremental settings can be found in~\cite{BrandCKLMGS24, BrandCKLPPSS24, ChenKLMP24}.

\subsection{Organization of the Paper}
We introduce relevant notations, definitions, and some key properties in \Cref{section:prelim}.
In \Cref{section:overview}, we provide an overview of our techniques.
Our results for max-flow sensitivity oracles are presented in \Cref{section:max-flow-sensitivity-oracle}.
In \Cref{section:strip-graph}, we review important structural properties of min-cuts that play a crucial role in computing min-cut sensitivity oracles.
The results for dual fault-tolerant and $k$-fault-tolerant min-cut sensitivity oracles are presented in \Cref{section:2-FT-oracle-using-FT-flows} and \Cref{section:FT-k-failures}, respectively.
We conclude with a discussion of open problems in \Cref{section:open-problems}.
Proofs omitted from the main text are provided in \Cref{section:deferred-proofs}, while our lower bound results are discussed in \Cref{section:lower_bounds}.
For completeness, we provide a discussion of the Circulation problem used in our max-flow sensitivity oracle in \Cref{section:circulation-with-lower-limits}.
In \Cref{section:baswana-min-cut-oracle-space}, we review the space complexity of \cite{BaswanaBP22} when parameterized by the min-cut size $\lambda$.

\section{Preliminaries}
\label{section:prelim}

Let $G = (V, E)$ be an unweighted directed multi-graph with $n$ vertices and $m$ edges. Hereon, we use the terms {\em graph} and {\em multi-graph} interchangeably. Let $s \in V$ be a source vertex, and $t \in V$ be a sink vertex. 
We use the notation $\maxflow(s,t,G)$, or equivalently $\mincut(s,t,G)$, to denote the value of the maximum-flow or minimum-cut in $G$ with respect to source $s$ and sink~$t$.
For any edge $e \in E$, we use $\overset{\leftarrow}{e}$ to denote the edge obtained by reversing the direction of $e$. 
For any vertex $x \in V$, we denote its set of out-neighbors by $\outv(x)$ and its set of in-neighbors by $\inv(x)$. 
For any subset of vertices $A \subseteq V$, the subgraph of $G$ induced by vertices in $A$ is denoted by $G[A]$. 
For any set $F$ of edges in $G$, the graph obtained by deleting the edges in $F$ from $G$ is denoted by $G - F$.

A set $C$ of edges in $G$ is called an $(s,t)$-cut if there is no $s$ to $t$ path in the graph $G - C$. 
Such a cut $C$ is called a {\em min-cut} if it minimizes the number of edges in the cut, i.e., it is the smallest set of edges whose removal disconnects $t$ from $s$. 
Any $(s,t)$-min-cut partitions the vertex set into two disjoint subsets: one containing the source vertex, denoted as $A_C$, and the other containing the sink vertex, denoted as $B_C$. Thus, an $(s,t)$-min-cut can equivalently be represented by this partition $(A_C, B_C)$.

Next we define the notion of nearest-cut, farthest-cut, $(\min+k)$-cut, and critical edges.
\begin{definition}
An $(s,t)$-min-cut $C^*$ is said to be the {Nearest Min-Cut}
$(${Farthest Min-Cut}$)$ if for each $(s,t)$-min-cut
$C$ we have $A_{C^*}\subseteq A_C$  $(A_{C^*}\supseteq A_C)$. 
We denote such a $C^*$ with $\nmc(s,t)$ $(\fmc(s,t))$.
\label{definition:NMC-FMC}
\end{definition}

\begin{definition}[$(Min+k)$-cut]
A set $\E$ of edges is said to be a $(\min+k)$ $(s,t)$-cut if $\E$ is an $(s,t)$-cut of size $k$ more than the size of minimum $(s,t)$-cut in $G$, and $\E$ is a {\bf minimal} $(s,t)$-cut (i.e., no proper subset of $\E$ is an $(s,t)$-cut).
\label{definition:min+k}
\end{definition}

\begin{definition}[Critical edge]
An edge $e$ in $G$ is said to be {\em critical} if it is contained in some $(s,t)$-min-cut of $G$. 
In other words, $e$ is a critical edge if it carries a unit flow with respect to every $(s,t)$-max-flow of $G$. 
Any edge which is not a critical edge is a non-critical edge. 
\label{definition:critical}
\end{definition}

For any flow $f$ from source $s$ to sink $t$ in a graph $G$, we use notation $\Null_{G}(f)$ to represent the collection of those edges in $G$ that carry zero flow under $f$. 
Further, $\Null_{G}(f,\min+1)$ represents the collection of those edges in $\Null_{G}(f)$ that 
are contained in some $(\min+1)$ $(s,t)$-cut in $G$. 

The following lemma (with proof deferred to the appendix, as it uses definitions from Section~\cref{section:strip-graph}) will be used crucially in our max-flow and min-cut sensitivity oracles.

\begin{lemma}
For any integer $(s,t)$-max-flow $f$ in an unweighted graph $G$, the cardinality of the set $\Null_{G}(f,\min+1)$ is at most $2n$.
\label{lemma:null-set-max-flow}
\end{lemma}

Below we state an assumption used in our max-flow and min-cut sensitivity oracles. 

\begin{assumption}
In our sensitivity oracles we assume that each vertex of the graph $G$ lies on some simple {\em $(s,t)$-path}. This assumption is justified as only edges incident to such vertices can contribute to an $(s,t)$-min-cut, even in the presence of edge failures.
\label{assumption-1}
\end{assumption}

\section{Technical Overview}
\label{section:overview}
In this section, we provide a detailed exposition of the core technical ideas of our paper. 

\subsection{Fault-tolerant Flow Family}
Let us begin by revisiting the $(s,t)$-reachability problem under single failure. It is well-known that for any two vertices $s$ and $t$ in a graph $G$ such that  $t$ is reachable from $s$, there exist two paths $P$~and~$Q$ intersecting {\em only} at $(s,t)$-cuts of size one. See Figure~\ref{fig:disjoint-paths-1}.
Consequently, if an edge~$e$ fails but does not disrupt $(s,t)$-reachability, at least one of these paths, $P$~or~$Q$, serves as a certificate of $(s,t)$-reachability. 

\medskip

\begin{figure}[!ht]
\centering
\includegraphics[width=0.54\linewidth]{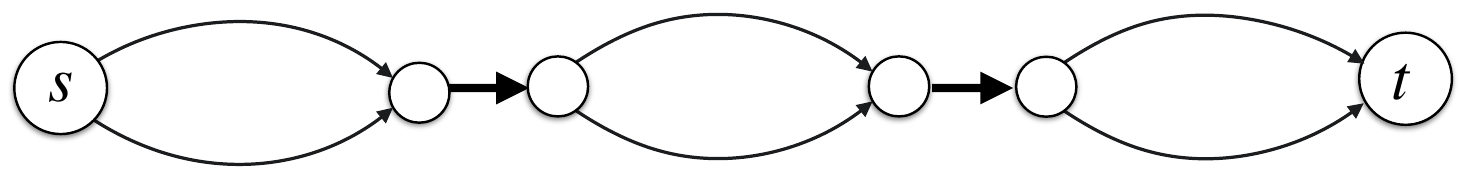}
\caption{A pair of two $(s,t)$ paths intersecting at $(s,t)$-cuts of size one.}
\label{fig:disjoint-paths-1}
\end{figure}

One of the main contributions of our work is extending this result to $(s,t)$-min-cuts of~size~$\lambda$ (greater than $1$). 
Specifically, we construct a family $\A$ of $\lambda + 1$ max-flows such that for any non-critical edge $e$ (i.e. $e$ satisfies the condition that failure of $e$ {\em does not} reduce the min-cut size), there exists a flow $f \in \A$ satisfying $f(e) = 0$. This flow serves as a max-flow certificate for the graph $G-e$. 

The construction for $\A$ proceeds through a sophisticated flow decomposition technique. First, an auxiliary graph $H = (V, E, c)$ is constructed where critical edges are assigned capacity $\lambda + 1$ while non-critical edges are assigned capacity $\lambda$. 
%
We argue that the $(s,t)$-max-flow $f$ in the capacitated graph $H$ has value $\lambda(\lambda + 1)$. 
To compute the family $\A=\{f_1, f_2, \ldots, f_{\lambda+1}\}$, we iteratively extract valid $(s, t)$-flows from $f$ in a total of $\lambda + 1$ iterations. 

These flows satisfy the property that for each edge $e$ in $H$, 
$$f(e) = f_1(e) + f_2(e) + \dots + f_{\lambda+1}(e).$$

To ensure that each iteration produces a valid maximum flow for $G$, we employ a careful application of circulation theory with lower bounds, where we put appropriate upper and lower bounds on the flow that can be removed from an edge during any iteration. 
We prove that these bounds guarantee that the extracted flow in each iteration is indeed a valid $(s, t)$-max-flow in $G$. 
%

Since $f = f_1 + f_2 +\cdots+ f_{\lambda+1}$ and the capacity of non-critical edges in $G$ is set to $\lambda$, it follows that for each non-critical edge $e$ there is at least one flow 
$f_i\in \A$
that satisfies $f_i(e) = 0$.

\subsection{Max-flow Sensitivity Oracle}

Building upon the fault-tolerant flow family, we develop compact sensitivity oracles for max-flow problem under both single and dual edge failures.

\paragraph{Single Edge Failure} 
Observe that a straightforward way to store edges carrying non-zero flow for each $f_i \in \A$ would require $O(\lambda^2 n)$ space, since each flow can saturate up to $\lambda n$ edges. However, we show that it is possible to store all flows implicitly in $O(\lambda n)$ space by carefully constructing $\mathcal{A}$ such that the number of edges on which any two flows differ is linear in $n$. To handle the failure of critical edges, the family $\A$ is further extended to compute a compact family $\B$ within the same size bounds, such that $\B$ contains a max-flow of $G-e$ for each edge $e$, whether critical or non-critical.

Our oracle supports two types of queries:
\begin{itemize}
\item Flow through specific edge: Given a failing edge $e$ and a query edge $x$, determine $f_e(x)$ in $O(1)$ time, where $f_e$ is a canonical max-flow of $G-e$.

\item Flow through all edges: Given a failing edge $e$, the oracle implicitly reports the updated max-flow in $O(n)$ time. 
\end{itemize}

\paragraph{Dual Edge Failure}
A natural extension of \Cref{thm-intro:FT-flow} is to ask whether the flow family can be generalized to handle multiple edge failures, thereby yielding max-flow sensitivity oracles for dual or higher failures. Unfortunately, this is not the case, as we prove an $\Omega(n)$ lower bound on the cardinality of $\B$ for $k \geq 2$ failures (for details see \Cref{theorem:FT-family-lower-bound}). 

However, we overcome this obstacle by exploiting the following key property:
for any $(s,t)$-flow $f$ in $G$ with $f(e) = 0$, the graph $G_f - e$ is the residual network of $(G-e)$ under $f$. Although seemingly straightforward, this observation provides a bridge between fault-tolerant flow families and our max-flow (as well as min-cut) sensitivity oracles under dual edge failures.

For instance, when edges $e, e'$ are deleted, the updated max-flow can be efficiently obtained if there exists a maximum flow $f$ satisfying $f(e) = 0$. In this case, re-routing reduces to adjusting the flow through $e'$ using graph $G_f - e$, rather than recomputing $(G-e)_f$. The existence of such flows are guaranteed by the family $\B$. Moreover, we show that by sparsifying $G_f$ using the 1-fault-tolerant strong connectivity preservers of Georgiadis et al.~\cite{GeorgiadisIP20}, this re-routing can be performed in $O(n)$ time.

\subsection{Reporting Min-Cut after Dual Failures~}

We next describe the construction of an oracle to determine the min-cut size after two edge failures.

{\bf Existing Approach.} ~As observed in~\cite{BaswanaBP22}, a challenging scenario arises when both failing edges are non-critical but their deletion reduces the min-cut size by one (i.e., the failing edges belong to a minimal cut of size $\lambda + 1$). To tackle this, \cite{BaswanaBP22} explicitly stored a family of $(\lambda + 1)$-sized minimal cuts in appropriate auxiliary graphs, in a total of $O(n^2)$ space.

{\bf Our Idea.} ~We take a detour from existing ideas and instead exploit the following insight based on max-flows: If an edge $e$ is non-critical, then there exists a flow $f$ such that $f(e) = 0$, making $f$ an $(s,t)$-max-flow in $G-e$. Then, upon deletion of another edge $e'$, the max-flow value decreases by one if and only if:
\begin{enumerate}
    \item $f(e') = 1$, and
    \item Flow through $e'$ cannot be rerouted via an alternate path,i.e., there is no cycle containing $(e')^{rev}$ in the residual graph $G_f-e$. See Figure~\ref{fig:scc-splitting}.
\end{enumerate}

\begin{figure}[!ht]
    \centering
    \includegraphics[width=0.9\linewidth]{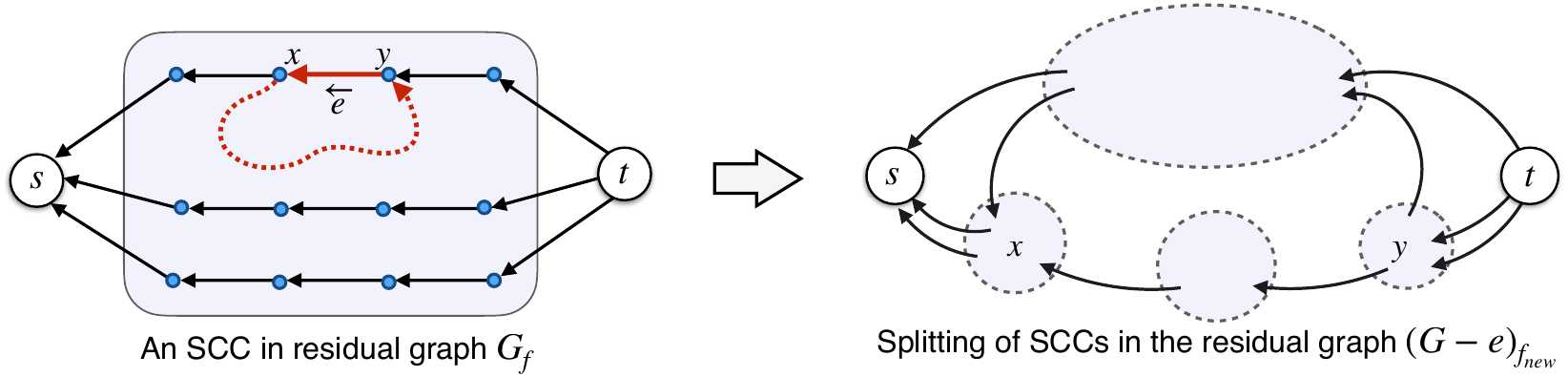}
    \caption{The SCCs in the residual graph $G_f$ may split on deletion of $e$, which makes some intra-SCC edges as inter-SCC in the updated residual graph.}
    \label{fig:scc-splitting}
\end{figure}

To exploit this observation, the first challenge is to bound the number of $(s,t)$-max-flows one will need. For this, the family $\A$ comes to our rescue and bounds the number of such flows by $\lambda+1$.
Another challenge is constructing an efficient and compact oracle to check if the above two conditions hold. We show that the first condition can be verified using our data structure over $\mathcal{A}$, while the second condition can be checked using a 1-fault-tolerant SCC data structure from \cite{GeorgiadisIP20} applied to the residual graph $G_f$. Both structures use linear space per flow in $\mathcal{A}$. This leads to an efficient $O(\lambda n)$-space data structure for min-cut sensitivity oracle for dual failures.
Our oracle to report the cut partitioning is more challenging and is deferred to the full version.

\subsection{Handling More than Two Edge Failures~}
We now discuss our ideas for handling general $k$ failures.
Given any set $F$ of failing edges of size at most $k$, we first consider the following simpler problem: Does the size of the $(s,t)$-min-cut 
decrease by exactly  $k$ when the edges in  $F$ fail?
This problem turns out to be relatively easier  because it requires checking if all $k$ edges simultaneously contribute to some $(s,t)$-min-cut in~$G$. 
We design a data structure ${\cal O}_{\text{MINCUT}}(s,t,G)$ of size $O(\lambda n)$ that verifies this in $O(k^2)$ time. Additionally, it reports a min-cut in the graph $G-F$ in $O(kn)$ time.

~Even though the simpler subproblem is easier to handle, the general problem of finding an $(s,t)$-min-cut in $G - F$ is significantly more challenging. 

For any failing set $F$ containing up to $k$ edges, it can be observed that:
$$|\mincut(s,t,G-F)| = \min_{0 \leq i < |F|}~ \min_{C \in (\lambda+i)\text{-cut}}~ (\lambda+i - |C \cap F|).$$

This suggests that studying minimal cuts of size up to $\lambda+k$ is crucial for answering min-cut queries. But these cuts can interact with each other in a number of ways, making it difficult to structurally analyse them. 

~To address the complexity of minimal cuts of size up to $\lambda+k$, we leverage a recent result by Kim et al.~\cite{KKPW-stoc22}, which provides a powerful graph augmentation technique. Specifically, their result shows that 
one can augment a graph with additional random edges $\E \subseteq V \times V$ so that a given minimal cut $C$ of size at most $\lambda+k$ becomes a valid minimum-cut in the augmented graph with probability $1/g(\lambda+k)$, for some computable function $g$. Intuitively, this augmentation allows us to transform the problem from studying minimal cuts to just minimum cuts. This augmentation, along with the data-structure ${\cal O}_{MINCUT}(s,t,G+\E)$ computed for each sampled set $\E$, helps us handle multiple failures effectively and obtain the desired oracle.

%

\section{A family of Fault-tolerant Flows}
\label{section:FT-flows}

Let $G=(V,E)$ be the input directed graph with unit edge capacities, $s$ be the source vertex, and $t$ be the sink. Let $\lambda$ be the value of $(s,t)$-max-flow in $G$. 
We present in this section construction of a compact family $\B$ of $2\lambda+1$ $(s,t)$-flows in $G$ with the property that, for each edge $e$ in $G$, there exists a maximum flow for the graph $G - e$ lying in family~$\B$. 


\smallskip

We will first establish the following theorem. 

\begin{theorem} \label{theorem:FT-flow-1}
Given any $n$-vertex directed graph $G=(V,E)$ with $(s,t)$-max-flow of value $\lambda$, we can compute in polynomial time a family $\A=\{f_1,\ldots,f_{\lambda+1}\}$ of $\lambda + 1$-max-flows such that for any non-critical edge $e\in E$, there exists a flow $f_i$ satisfying $f_i(e)=0$.
\end{theorem}

We compute an auxiliary graph $H=(V,E,c)$, where for each edge $e\in E$, $c(e)$ is defined as follows: 

$$
c(e)=
\begin{cases}
\lambda+1 ~~& \text{if $e$ is critical in $G$,}\\
\lambda & \text{otherwise.}
\end{cases}
$$

\vspace{2mm}

\begin{lemma}
\label{claim:H-max-flow-value}
$\maxflow(s,t,H)=\lambda (\lambda + 1)$.
\end{lemma}
\begin{proof}
Consider an $(s, t)$-min-cut $C$ in $G$. As each edge in $C$ has a capacity of $\lambda + 1$ in $H$, the capacity of this cut in $H$ is $\lambda (\lambda + 1)$. So, the $(s,t)$-min-cut value in $H$ is at most $\lambda (\lambda + 1)$.

Next, observe that any $(s,t)$-cut $C$ that is not a min-cut in $G$  must contain at least $\lambda + 1$ edges. Since each edge in $H$ has a capacity at least $\lambda$, the capacity of any cut $C$ that is not a min-cut in $G$ is at least $\lambda (\lambda + 1)$ in $H$. This proves that 
$\maxflow(s,t,H)=\lambda (\lambda + 1)$.
\end{proof}

\vspace{0.5mm}

Let $f$ be integral $(s,t)$-max-flow of value $\lambda (\lambda+1)$ in $H$. The main idea behind computing the family $\A$ is to iteratively peel off from $f$, $\lambda$ units of flow, in a total of $(\lambda+1)$ rounds. 
In order to proceed, we present the following lemma.

\vspace{0.5mm} 

\begin{lemma}
\label{lemma:flow-existence}
Let $h$ be an integral $(s,t)$-flow in $H$ of value $\lambda i$ (assuming $i \in \mathbb{Z}^+$) such that  $h(e) \leq i$, for each $e \in E$. 
Then we can compute a max-flow, say $f_i$, in $G$ with 0/1 values satisfying

$$
f_i(e)=
\begin{cases}
    0 & \text{~if~~} h(e)=0,\\
    1 & \text{~if~~} h(e)=i.
\end{cases}
\quad\text{for~each~edge~}e\in E.
$$
\end{lemma}

\begin{proof}
To compute the flow $f_i$,
we construct an instance $\widetilde G=(V,E,d,\ell,\mu)$ of the circulation problem (see Appendix~\ref{section:circulation-with-lower-limits}) with upper and lower limits as described below:
\vspace{1mm}
\begin{enumerate}
\item $d(s) = -\lambda$, $d(t) = \lambda$ and $d(v) = 0 \ \forall v \neq s, t$
\item For each $e\in E$, $\mu(e)=\min\{1,h(e)\}$
\item For each $e\in E$, $\ell(e)=1$ if $h(e)=i$, and $\ell(e)=0$ otherwise.
\end{enumerate}

\noindent
It suffices to show that $\widetilde G$ has a circulation. 
By \Cref{lemma:circulation-with-lower-limits} in appendix, a circulation exists in $\widetilde G$ if and only if, for every cut $(A,B)$ in $\widetilde G$, the following inequality holds.
$$
d(B)+\ell(B,A)~\leq~ \mu(A,B),
$$
where, 
$$
d(B)=\sum_{v\in B} d(v)
\text{,~~~~} 
\ell(B,A)=\displaystyle \sum_{\substack{(x,y)\in E,\\x\in B,~y\in A}}\ell(x,y)
\text{,~~~and~~~} 
\mu(A,B)=\displaystyle \sum_{\substack{(x,y)\in E,\\x\in A,~y\in B}}\mu(x,y).
$$
\vspace{2mm}

Consider a cut $(A,B)$ in $H$. Let $\mathbbm{1}_{t\in B}$ be an indicator variable that takes 
the value $1$ if $t$ lies in $B$, and $0$ otherwise. Similarly, define 
$\mathbbm{1}_{s\in B}$. So, the net flow under $h$ from $A$ to $B$ is 
$$\lambda i(\mathbbm{1}_{t\in B}-\mathbbm{1}_{s\in B}).$$ 

Let $\alpha$ be the number of edges in $H$ lying in set $B\times A$ that carry a flow $i$ under $h$. 
Thus, the edges in $H$ lying in $A\times B$ must be carrying a flow of at least
$$(\alpha i)~+~(\lambda i)(\mathbbm{1}_{t\in B}-\mathbbm{1}_{s\in B}).$$

As each edge in $H$ carries a flow at most $i$ under $h$, the number of edges in the set $A\times B$ that were carrying 
a non-zero flow with respect to $h$ must be at least $\alpha+\lambda (\mathbbm{1}_{t\in B}-\mathbbm{1}_{s\in B})$.
This implies $$\mu(A,B)\geq \alpha+\lambda (\mathbbm{1}_{t\in B}-\mathbbm{1}_{s\in B}).$$

Next recall that there are $\alpha$ edges in $H$ from $B$ to $A$ that carry a flow $i$ under $h$, and for each such edge $e$, we have $\ell(e)=1$. 
Thus, $\ell(B,A)=\alpha$. Further, $d(B)=\lambda (\mathbbm{1}_{t\in B} - \mathbbm{1}_{s\in B})$. 
This establishes that $d(B)+\ell(B,A)~\leq~ \mu(A,B)$, for every cut $(A,B)$ in $\widetilde G$. Therefore, a circulation exists in $\widetilde G$. 
\end{proof}

\begin{lemma}
We can compute in polynomial time a family $\A=\{f_1,\ldots,f_{\lambda+1}\}$ of $(\lambda+1)$ integral max-flows in $G$ satisfying $f=f_1+\cdots+f_{\lambda+1}.$
\label{lemma:f-decomposition}
\end{lemma}
\begin{proof}
We begin with the $(s,t)$-max-flow $f$ in $H$, which satisfies the conditions of \Cref{lemma:flow-existence} for $i = \lambda + 1$. So, $f(e) \leq \lambda + 1$ for all edges $e \in E$, and value of $f$ is $\lambda(\lambda + 1)$. On applying \Cref{lemma:flow-existence} to flow $h_{\lambda+1} := f$ in $H$, we obtain the flow $f_{\lambda+1}$.

This process is then repeated for each $i$ from $\lambda$ down to 1. For each $i$, we compute $f_i$ by considering the flow $h_i := f - (f_{i+1} + \cdots + f_{\lambda+1})$ in the graph $H$. It is important to note that $0 \leq h_i(e) \leq i$ for all edges $e \in E$. This is due to the fact that, in any previous invocation of \Cref{lemma:flow-existence}, say during round $j$ (where $j > i$), if $h_j(e) = j$, then $f_j(e) = 1$; similarly, if $h_j(e) = 0$, then $f_j(e) = 0$. Given that value of $h_i$ is $\lambda i$, we can compute $f_i$ by applying \Cref{lemma:flow-existence} to the flow $h_i$ in $H$.
This ultimately results in a sequence of $\lambda + 1$ integral max-flows, $f_1, f_2, \dots, f_{\lambda + 1}$, in $G$  which satisfy the equation
$f = f_1 + f_2 + \cdots + f_{\lambda + 1}.$
\end{proof}

\vspace{.5mm}

Consider the family $\mathcal{A} = \{ f_1, \dots, f_{\lambda+1} \}$ obtained from \Cref{lemma:f-decomposition}. We now show that for any non-critical edge $e$, there exists a flow $f_i \in \mathcal{A}$ such that $f_i(e) = 0$. 
Consider a non-critical edge $e$ in $G$. Recall that, by the construction of $H$, we have $c(e) = \lambda$. Since $f(e) \leq c(e)$, it follows that
$$
f_1(e) + \cdots + f_{\lambda+1}(e) \leq \lambda.
$$
Thus, there must exist some flow $f_i \in \mathcal{A}$ such that $f_i(e) = 0$. 
This concludes the proof of \Cref{theorem:FT-flow-1}.

\paragraph{Construction of family $\B$}

We next extend our construction to obtain the following result.

\begin{theorem}
For any graph $G$ with $(s,t)$-max-flow~$\lambda$, we can compute in polynomial time a family $\B$ 
of $2\lambda + 1$ $(s,t)$-flows in $G$ satisfying the following:
\begin{center}
For each edge $e$, the family $\B$ contains an $(s,t)$-max-flow for $G - e$.    
\end{center}
\label{theorem:FT-flow-2}
\end{theorem}
\vspace{-4mm}

\begin{proof}
Let $\A = \{ f_1, \ldots, f_{\lambda+1} \}$ be the family of flows obtained from \Cref{theorem:FT-flow-1}. Then, for any non-critical edge $e$, there exists a flow $f_i \in \A$ such that $f_i(e) = 0$.

Next, fix an $(s,t)$-flow $\tilde{f} \in \A$. Let $\P = \{ P_1, \ldots, P_\lambda \}$ be a family of $\lambda$ edge-disjoint $(s,t)$-paths satisfying that for each edge $e \in E$, we have $\tilde{f}(e) = 1$ if and only if $e$ is contained in one of the paths in $\mathcal{P}$. 
Define a collection of $\lambda$ flows, $g_1, \ldots, g_\lambda$, where each flow $g_i$ is obtained by cancelling the flow along path $P_i$ from $\tilde{f}$. For any critical edge $e$, if $P_i \in \mathcal{P}$ is the path containing $e$, then $g_i$ is an $(s,t)$-max-flow for the graph $G - e$. 

Thus, the family 
$\A\cup \{g_1, \ldots, g_\lambda \}$ consists of the required set of $(s,t)$-flows in $G$.
\end{proof}

We remark that the bound of $(2\lambda + 1)$ on the size of the family $\mathcal{B}$ in \Cref{theorem:FT-flow-2} is existentially tight.
Consider any graph $G$ having an $(s,t)$ minimal-cut of size $(\lambda + 1)$, say $C_{\lambda+1}$, that does not contain any critical edges. 
(A~simple construction for such a $G$ is a graph  with three vertices $s$, $x$, $t$, with $\lambda$ parallel edges from $s$ to $x$, and $\lambda + 1$ parallel edges from $x$ to $t$). Additionally, let $C_\lambda$ be an $(s,t)$-minimum cut in $G$.
For each potential failure $e \in C_\lambda$, there exists an $(s,t)$-max-flow in $G - e$ of value $\lambda - 1$, which saturates all edges in $C_\lambda - e$. Similarly, for each failure $e \in C_{\lambda+1}$, there exists an $(s,t)$-max-flow in $G - e$ of value $\lambda$, which saturates all edges in $C_{\lambda+1} - e$. It is easy to observe that these flows must be distinct. 
This provides us a lower bound of $(2\lambda + 1)$ on the size of family $\B$.

\section{Maximum-flow Sensitivity Oracle} 
\label{section:max-flow-sensitivity-oracle}

In this section, we present our sensitivity oracle for reporting the maximum flow in the presence of single and dual edge failures. Let $G = (V, E)$ be an $n$-vertex directed graph with source $s$ and sink~$t$, and let $\lambda$ be the value of the $(s, t)$-max-flow in $G$. 

We begin by transforming $G$ into another graph $\G=(V,\E)$ as follows.
Initialize $\G$ as $G$, and next iteratively remove those edges $e$ from $\G$ that do not lie in any minimal-cut of size
$\lambda$~or~$\lambda+1$ in~$\G$. 

\begin{observation}
$\G$ is a minimal subgraph of $G$ satisfying the following:
\begin{enumerate}
\item $\maxflow(s,t,G)=\maxflow(s,t,\G)$, and
\item $\maxflow(s,t,G-e)=\maxflow(s,t,\G-e)$, for each $e\in E(G)$.
\end{enumerate}
\end{observation}


Throughout this section, we use $\A$ and $\B$ to denote the families of flows obtained by applying \Cref{theorem:FT-flow-1} and \Cref{theorem:FT-flow-2} to the graph $\G$. 

\smallskip

So, for any flow $f\in \B$, $\Null_{\G}(f)$ represents the collection of those edges in $\G$ that carry zero flow under $f$. The lemma below bounds the size of set $\Null_{\G}(f)$, for any flow $f\in \B$.

\begin{lemma}
For any flow $f\in \B$, the cardinality of the set $\Null_{\G}(f)$ is at most $3n$.
Furthermore, the number of edges in $\G$ is at most $O(\lambda n)$.
\label{lemma:null-set-f-edge-set}
\end{lemma}

\begin{proof}
Consider any max-flow $f \in \A$. By \Cref{lemma:null-set-max-flow}, the size of the set $\Null_\G(f, \min+1)$ is at most $2n$. Observe that each edge in $\G$ is contained in either an $(s,t)$-min-cut or an $(s,t)$-$(\min+1)$-cut in $\G$. Thus, $\Null_\G(f) =\Null_\G(f, \min+1)$, as the critical edges in $\G$ cannot carry zero flow. This proves a bound of $2n$ on the cardinality of the set $\Null_\G(f)$.

Now, let $\tilde{f}$ be the representative flow in $\A$ used to compute the flows $g_1, \ldots, g_\lambda \in \B \setminus \A$. Then each flow in $\B \setminus \A$ is obtained by cancelling one-unit flow along an $(s,t)$-path in $\G$ from~$\tilde{f}$. 
%
%
For each $g_i$, it follows that $\Null_\G(g_i)$ differs from $\Null_\G(\tilde{f})$ in at most $n$ edges, since cancellation of flow along a single path changes the number of edges in the $\Null$ set by at most $n$. Given that $|\Null_\G(\tilde{f})|$ is bounded by $2n$, we conclude that $|\Null_\G(g_i)|$ is bounded by $3n$, for each $i \in [1, \lambda]$.

We now prove the second part. The number of critical edges in $\G$ is at most $\lambda n$, as the value of the $(s,t)$-max-flow in $\G$ is $\lambda$. Since each non-critical edge in $\G$ is contained in an $(s,t)$-$(\min+1)$-cut, the non-critical edges are contained in the union $\bigcup_{f \in \A} \Null_\G(f,\min+1)$. Hence, the total number of edges in $\G$ is bounded by $O(\lambda n)$.
\end{proof}

As a corollary of \Cref{lemma:null-set-f-edge-set}, note that the following extension of \Cref{theorem:FT-flow-2} is immediate.

\begin{theorem}
For any graph $G$ with $(s,t)$-max-flow~$\lambda$, we can compute in polynomial time a family $\B$ 
of $2\lambda + 1$ $(s,t)$-flows in $G$ satisfying the following:
\begin{enumerate}
\item For each edge $e$, the family $\B$ contains an $(s,t)$-max-flow for $G - e$.    
\item For any two flows $f,f'\in \B$, the flows $f$ and $f'$ disagree on at most $O(n)$ edges.
\end{enumerate}
\label{theorem:FT-flow-3}
\end{theorem}

\subsection{Reporting Maximum Flow under Single Edge Failure}
\label{subsection:max-flow-edge-failure-ds}

We now present a sensitivity oracle for answering $(s,t)$-max-flow queries in the presence of a single edge failure. Our oracle stores the following:

\begin{itemize}
\item A dictionary of the edges in $E(\G)$. Additionally, for each flow $f \in \B$, the oracle stores a dictionary of the set $\Null_{\G}(f)$.
\item Label of a canonical $(s,t)$-maximum flow of the graph $\G$, denoted $\tilde{f}$, where $\tilde{f} \in \A$.
\item For each edge $e \in E(\G)$, the label of a canonical flow $f_e \in \B$ such that $f_e$ is an $(s,t)$-max-flow in $\G - e$.
\end{itemize}

Due to \Cref{lemma:null-set-f-edge-set}, the cardinality of the set $\Null_{\G}(f)$ is at most $3n$, for any $f\in \B$; and the number of edges in $\G$ is at most $O(\lambda n)$. Therefore, the size of the data structure is $O(\lambda n)$.

Before presenting the query algorithm, we state the following observation.

\begin{observation}
For any edge $e\in E(G)$ and any flow $f\in \B$, the flow through $e$ under $f$ can be determined in $O(1)$ time.
\label{observation:flow-query}
\end{observation}
\begin{proof}
Consider a flow $f\in \B$.
Observe that $f(e)=1$ if and only if $e$ lies in $\G$ but is not contained in $\Null_{\G}(f)$. Since we store dictionaries for both $E(\G)$ and $\Null_{\G}(f)$, querying whether $e$ lies in $E(\G)\setminus \Null_{\G}(f)$ requires just $O(1)$ time.
\end{proof}

\begin{figure}[!htt]
\centering
\begin{minipage}{\textwidth}
\centering
\begin{minipage}{0.49\textwidth}
\begin{algorithm}[H]
\caption{\textsc{Query-Edge-Flow}$(e,x)\hspace{-6mm}$}
\setstretch{1.1}
\tcp{Query flow through edge $x$ after failure of $e$}
\vspace{2mm}
\uIf{$\tilde{f}(e)=0$}
    {Return $\tilde{f}(x)$}
\Else{
    Let $f_e$ be the canonical flow in $\B$ satisfying $f_e$ is max-flow in $\G-e$\; 
    Return $f_e(x)$\;
}
\label{alg:ReportFlow}
\end{algorithm}
\end{minipage}
\hfill
\begin{minipage}{0.49\textwidth}
\begin{algorithm}[H]
\caption{\textsc{Report-Flow-Diff}$(e)\hspace{-6mm}$}
\setstretch{1.1}
\tcp{Report edges whose flow is altered after failure of $e$}
\vspace{2mm}
\uIf{$\tilde{f}(e)=0$}
    {Return $\emptyset$}
\Else{
    Let $f_e$ be the canonical flow in $\B$ satisfying $f_e$ is max-flow in $\G-e$\; 
    Return $\Null_{\G}(\tilde{f}) \oplus \Null_{\G}(f_e)$\;
}
\label{algo:ReportDiff}
\end{algorithm}
\end{minipage}
\end{minipage}
\end{figure}

Our max-flow sensitivity oracle supports the following two natural operations upon failure of an edge in the graph.

\paragraph{Query flow through a given edge $x$}~\\
For a failing edge $e$, if $\tilde{f}(e) = 0$ then there is no change in the flow, so we simply return $\tilde f(x)$. Otherwise, we use the precomputed flow $f_e$ (a max-flow in $\G-e$) to return $f_e(x)$. Both these steps require only $O(1)$ time due to \Cref{observation:flow-query}. See \Cref{alg:ReportFlow}.

\paragraph{Report the set of edges whose flow value changes}~\\
Consider a failing edge~$e$ in $G$. If $\tilde f(e)=0$, then there is no change in the flow, so we simply return the empty-set.
If not, we retrieve the label of flow $f_e \in \B$ that is an $(s, t)$-max-flow in $\G - e$. Next, to identify the edges over which $\tilde f$ and $f_e$ differ, we simply traverse the sets $\Null_{\G}(\tilde f)$ and $\Null_{\G}(f_e)$, and output their symmetric difference. This operation takes $O(n)$ time, as by \Cref{lemma:null-set-f-edge-set}, the size of both the sets involved is at most $3n$. 
See \Cref{algo:ReportDiff}.

\bigskip

\noindent
This completes our discussion on the max-flow sensitivity oracle for single failure.

\subsection{Reporting Maximum Flow under Dual Edge Failures}

We now address the problem of designing a sensitivity oracle for reporting the $(s,t)$-max-flow after the failure of any two edges. The foundation of our approach is the fault-tolerant flow family~$\B$ and 1-FT-SCC certificate of \cite{GeorgiadisIP20}.  

Before presenting the query procedure, we state some structural observations, which are essential for efficiently reporting flows.

\begin{observation}
\label{obs:1-FT-flow}
Suppose $H$ is a directed graph, $f$ is an $(s,t)$-max-flow in $H$, and $e$ is an edge with $f(e)=1$. Then the flow through $e$ can be rerouted along an alternate path if and only if there exists a cycle in the residual graph $H_f$ that contains $e^{\rev}$.
\end{observation}

\begin{observation}
For any graph $H$ with $(s,t)$-max-flow $f$ and any edge $e$ satisfying $f(e) = 0$, the residual graph $(H-e)_f$ is same as $H_f - e$.
\label{obs:H-e-residual}
\end{observation}

Given a cycle $C$ in $H_f$, we define $K(C)$ to be the collection of edges in $G$ corresponding to $C$. More specifically, $K(C) = (E(C)\cap E(G)) \cup \{e^{\rev}\mid e\in E(C)\setminus E(G)\}$. Switching flow through a cycle $C$ in the residual graph corresponds to toggling the set of saturated edges along $K(C)$. 

\paragraph{Query Algorithm}
Given a pair of failing edges $e$ and $e'$, the query algorithm proceeds as follows. First, we retrieve a canonical flow $f$ in $\B$ that is a maximum $(s,t)$-flow in $G-e$. If this flow already avoids both failures, i.e., $f(e)=0$ and $f(e')=0$, then it is also a maximum flow in $G-\{e,e'\}$, and so the change in the edges carrying flow is simply $X = \Null_{\G}(\tilde f) \oplus \Null_{\G}(f)$.

Now, suppose $f(e') = 1$. By \Cref{obs:H-e-residual}, we have $G_f - e$ is the residual graph of $G-e$ with respect to $f$. We next determine if there exists a directed cycle in $G_f-e$ containing $(e')^{\rev}$, if so, we can reroute the flow along this cycle. 

If, however, no such cycle exists, then flow through $e'$ cannot be rerouted and its endpoints must belong to different strongly connected components in $G_f-e$.
In this situation, we must reduce the total flow by one unit, which is achieved by canceling a unit of flow along an $(s,t)$ path that traverses $e'$. To accomplish this, we take the help of graph $G_f + (s,t) -e$  (where a direct edge $(s, t)$ is artificially introduced). 
In this setting, the desired $(s,t)$ path containing $e'$ corresponds to a cycle in $G_f + (s, t) - e$ containing edge $(e')^{\rev}$. 
See~\Cref{alg:ReportMaxFlow-Dual}.

\begin{algorithm}[!ht]
\setstretch{1.35}
\caption{\textsc{Report-Flow-Diff}($e, e'$)}
Let $f$ be a flow in $\B$ maximizing $val(f)$ such that $f(e)=0$\;
Compute $X=\Null_{\G}(\tilde f)\oplus \Null_{\G}(f) $\;
\lIf{$f(e')=0$}{return $X$}
\lElseIf{$\exists$ a cycle $C_1$ in $G_f-e$ containing  $(e')^{rev}$}{
return $X\oplus K(C_1)$}
\lElse{
find a cycle $C_2$ in $G_f+(s,t)-e$ containing  $(e')^{rev}$, and return $X \oplus K(C_2)$
}
\label{alg:ReportMaxFlow-Dual}
\end{algorithm}


\paragraph{Use of SCC Preservers}
To realize the query algorithm efficiently, the only additional structure we need is, for each $f \in \B$, a way to detect cycles in the residual graphs $G_f - e$ and $G_f + (s,t) - e$. This is achieved by storing, for each $f$, the 1-fault-tolerant SCC preserver of Georgiadis et al.~\cite{GeorgiadisIP20}.
which is a subgraph with $O(n)$ edges maintaining all strongly connected components under single-edge failures. With this, checking for the existence of cycles passing through an edges takes $O(n)$ time; and this does not inflate the overall space usage.

\begin{theorem}[Georgiadis et al.~\cite{GeorgiadisIP20}]
For any directed graph $H$, there exists a subgraph $H_0$ on $O(n)$ edges such that for any edge $e \in E(H)$, the strongly connected components of $H-e$ and $H_0-e$ are identical.
\label{theorem:1-ft-scc-preserver}
\end{theorem}

The data structure thus consists of (a) dictionaries for $E(\mathcal{G})$ and sets $\Null_\G(f)$ for $f \in \mathcal{B}$, and (b) the 1-FT-SCC preservers for $G_f$ and $G_f+(s,t)$ for all $f$. Since $|\mathcal{B}| = O(\lambda)$, the total space remains $O(\lambda n)$. 
We now discuss the time complexity of  \Cref{alg:ReportMaxFlow-Dual}.
Identifying the flow $f$ just takes $O(1)$ time. Then, finding the cycle in the appropriate SCC preserver further takes $O(n)$ time, as number of edges in any FT SCC-preserver is $O(n)$. Hence, the total time to report all edges where the flows differ is $O(n)$.

\begin{theorem}
There exists an $O(\lambda n)$-space oracle that, for any pair of failing edges $e, e'$, efficiently reports all edges whose flow value is altered in a maximum $(s,t)$-flow after failure of $e$ and $e'$.
\end{theorem}

\section{An Overview of Structural Properties of (s,t)-Min-Cuts}
\label{section:strip-graph}
In this section, we discuss some structural properties of $(s,t)$-min-cuts that will be essential for the computation of min-cut sensitivity oracles.

Consider a relation on vertex-set of $G$ under which any two vertices $x$ and $y$ are said to be related if and only if they are not separated by any $(s,t)$-min-cut in $G$. Let $\W$ be the collection of equivalence classes of $V$ induced by this relation, and for any $v\in V$, let $\bm v$ denote the equivalence class of $v$ in~$G$. We say an edge $(x,y)$ in $G$ is {\em inter-cluster} if $\bm x\neq \bm y$, and {\em intra-cluster} otherwise. 

Observe that a critical edge is always inter-cluster; however, in directed graphs the converse is not necessarily true (see~\Cref{fig:G}).

\begin{figure}[htp]
    \centering
    \fbox{
        \begin{minipage}{0.98\textwidth} 
            \centering
            \subfigure[Directed graph $G$ with an $(s,t)$-max-flow $f$. 
            Dashed edges in red are inter-cluster non-critical.]{
                \includegraphics[width=0.43\textwidth, trim=1mm 2mm 3mm 2mm, clip]{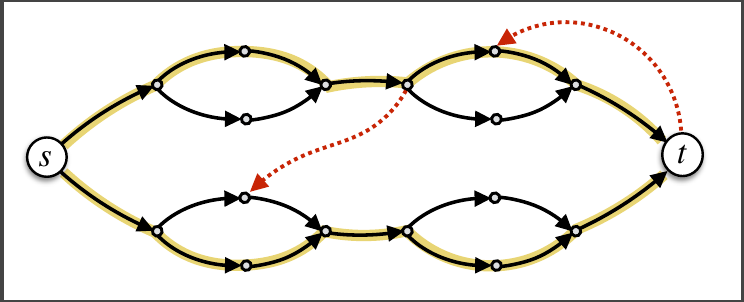}
                \label{fig:G}
            }
            \quad~
            \subfigure[Quotient graph $G_\lambda$ obtained by condensing all the equivalence classes in $\W$ into supernodes.]{
                \includegraphics[width=0.43\textwidth, trim=2mm 2mm 3mm 2mm, clip]{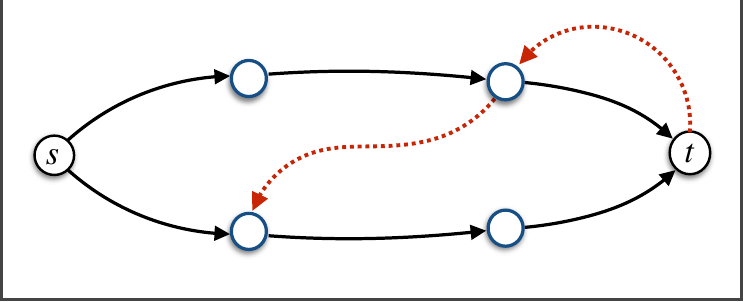}
                \label{fig:G_lambda}
            }
            \\[2mm]
            \subfigure[Residual graph $G_f$. The SCCs of $G_f$ coincide 
            with equivalence classes in $\W$.]{
                \includegraphics[width=0.43\textwidth, trim=1mm 2mm 3mm 2mm, clip]{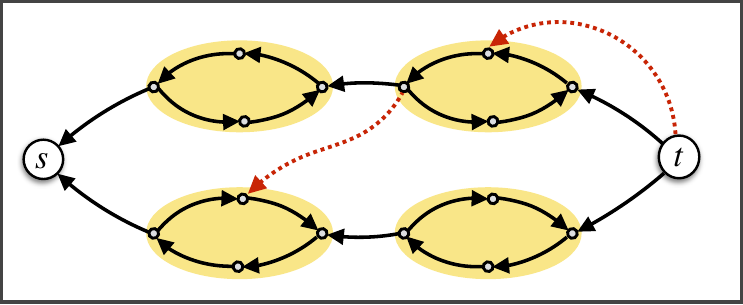}
                \label{fig:G_f}
            }
            \quad~
            \subfigure[Strip graph~$D_\lambda$ obtained by reversing direction 
            of non-critical edges in $G_\lambda$. ]{
                \includegraphics[width=0.43\textwidth, trim=2mm 2mm 3mm 2mm, clip]{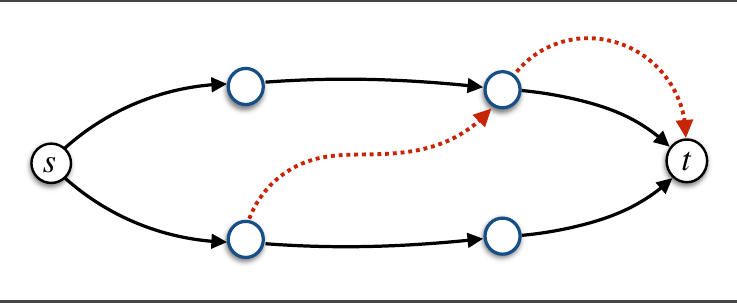}
                \label{fig:D_lambda}
            }
        \end{minipage}
    }
    \caption{Depiction of graphs $G_\lambda$ and $D_\lambda$. 
    Note that the strip graph $D_\lambda$ is precisely the DAG of the SCCs of $G_f$, but with all edge directions reversed.}
    \label{fig:GD_lambda}
\end{figure}

Throughout this paper, we use $G_\lambda$ to denote the quotient graph of $G$ induced by the relation defined above, and $D_\lambda$ to refer to the graph obtained by reversing the direction of all non-critical inter-cluster edges in $G_\lambda$,
where $\lambda$ denotes the value of the $(s,t)$-max-flow in the graph $G$.

Picard and Queyranne~\cite{PicardQ82}, Dinitz and Vainshtein~\cite{DinitzV00} referred to $D_\lambda$ as {\em strip graph} and presented the following properties of~$D_\lambda$ (under \Cref{assumption-1}). 

\begin{property}[\hspace{-0.1mm}\cite{PicardQ82}]
For any $(s,t)$-max-flow $f$ in~$G$, the SCCs of residual graph $G_f$ correspond to the equivalence classes $\W$ of $G$. Moreover, $D_\lambda$ is essentially the DAG obtained by reversing the edges of graph obtained by contracting the SCCs of $G_f$ into supernodes.
\label{property:DLambda-construction}
\end{property}

\Cref{fig:GD_lambda} presents a depiction of the graphs $G_\lambda$ and $D_\lambda$, along with the alternate characterization of  $D_\lambda$	 as described in the above property.

\begin{property}[\hspace{-0.1mm}\cite{PicardQ82,DinitzV00}]
An $(s,t)$-cut $C$ in $G$ is a minimum cut if and only if  $C$~comprises of critical inter-cluster edges, and  the edges of the cut intersect any path in $D_\lambda$ at most once. 
\label{property:transversal}
\end{property}

\section{Dual Fault-Tolerant Min-cut oracle via Fault-Tolerant Flows}
\label{section:2-FT-oracle-using-FT-flows}

We present here a construction of a dual fault-tolerant min-cut oracle that, for any graph $G$ with an $(s,t)$-min-cut of size $\lambda$, uses $O(n\lambda)$ space and can report the size of the min-cut upon the occurrence of failures in constant time.

In the first subsection, we show how to handle failing sets that do not contain any critical edges by employing the fault-tolerant flow family $\A$ computed in \Cref{theorem:FT-flow-1}. In the subsequent subsection, we discuss how to handle failing sets that contain one or more critical edges.

\subsection{Handling failure of non-critical edges}
\label{section:min-cut-no-critical-edges}

Let $\A=\{f_1, \dots f_{\lambda+1}\}$ be collection of $(\lambda+1)$ $(s,t)$-max-flows obtained by applying \Cref{theorem:FT-flow-1} on graph $G$. 

\begin{lemma}
Let $e_1,e_2\in E$ and $f\in \A$ be a max-flow such that $f(e_1)=0$ and $f(e_2)=1$. Then,
$\maxflow(s,t,G-\{e_1,e_2\})=\lambda-1$ if and only if the endpoints of $e_2$ are not strongly connected in $G_{f}-e_1$, where $G_f$ is the residual graph corresponding to flow $f$.
\label{lemma:min-cut-using-SCC-oracle}
\end{lemma}
\begin{proof}
Since $\maxflow(s,t,G-e_1)=\lambda$, the size of $(s,t)$-min-cut in $G-\{e_1,e_2\}$ is $\lambda-1$ if and only if $e_2$ is a critical edge in $G-e_1$. 

Observe that critical edges in a graph are those inter-cluster edges that are saturated with respect to every $(s,t)$-max-flow. Likewise, non-critical inter-cluster edges are those inter-cluster edges that carry a zero flow with respect to every $(s,t)$-max-flow. Indeed, if $(A,B)$ is an $(s,t)$-min-cut separating endpoints of a non-critical edge $e$, then $e$ would be directed from set $B$ to set $A$. Since edges directed from sink-side to source-side of an $(s,t)$-min-cut carry a zero flow with respect to every max-flow, we must have $f(e)=0$, for every $(s,t)$-max-flow $f$.
 
Since $f(e_2)=1$, edge $e_2$ is a critical edge in $G-e_1$ if and only if $e_2$ is inter-cluster edge in $G-e_1$, that is, the endpoints of $e_2$ lie in different SCCs in the residual graph $(G-e_1)_{f}=G_f-e_1$. This proves the claim. 
\end{proof}

In order to use the above lemma to answer min-cut queries we need an efficient data structure for strong-connectivity upon edge failures in residual graphs. For this, we use the following result by Georgiadis, Italiano, and Parotsidis~\cite{GeorgiadisIP20}.

\begin{lemma}[Georgiadis et al.~\cite{GeorgiadisIP20}]
\label{lemma:fault-tolerant-scc}
For any $n$ vertex directed graph $G = (V, E)$, there exists an $O(n)$ sized data structure ${\cal O}_{SCC}(G)$ that, given any pair of vertices $x,y \in V$ and an edge $e \in E$, answers in $O(1)$ time whether or not $x$ and $y$ are strongly connected in $G - e$. 
\end{lemma}

We now describe the construction of oracle that answers min-cut queries on failure of two non-critical edges. Our oracle stores the following information.

\begin{itemize}
\item The fault-tolerant strong-connectivity oracle ${\cal O}_{SCC}(G_f)$ of graph $G_f$, for each $f\in \A$.
\item For each non-critical edge $e$ contained in a $(\min+1)$-cut, 
label of a flow in $\A$, denoted $f_e$, under which edge $e$ carries zero flow.
\item A dictionary of the set $\bigcup_{f\in \A} \Null_G(f,\min+1)$. Additionally, for each flow $f \in \A$, the oracle stores a dictionary of the set $\Null_{\G}(f,\min+1)$.
\end{itemize}

It is easy to verify that the oracle takes $O(\lambda n)$ space.

\paragraph{Query Oracle}
Consider a pair of non-critical edges, $e$ and $e'$. We first verify whether $e$ as well as $e'$ are contained in some $(\min+1)$-cut, i.e., lie in the union $\bigcup_{f \in \A} {\Null}_G(f, \min+1)$. Checking their membership in this union suffices as for each non-critical edge, there exists a flow $f\in \A$ under which the edge carries zero flow.

If either $e$ or $e'$ is not contained in $\bigcup_{f \in \A} {\Null}_G(f, \min+1)$, the min-cut size remains unchanged, as there will be no $(\lambda+1)$-minimal cut containing these edges.

Next, assume both $e$ and $e'$ are contained in $\bigcup_{f \in \A} {\Null}_G(f, \min+1)$. Recall $f_e \in \A$ denotes the max-flow under which $e$ carries zero flow. By \Cref{lemma:min-cut-using-SCC-oracle}, the $(s,t)$-min-cut size decreases by one if and only if the endpoints of $e'$ are in different strongly connected components (SCCs) in $G_{f_e} - e$, and $f_e(e') = 1$. 

The first condition, which checks whether the endpoints of $e'$ belong to different SCCs, can be verified in $O(1)$ time using the $\mathcal{O}_{SCC}(G_{f_e})$ oracle. The second condition, $f(e') = 1$, is verifiable by checking that $e'$ is not present in the dictionary of ${\Null}_G(f_e, \min+1)$, a check that can also be performed in constant time. 
Thus, the query time of the oracle is $O(1)$. 


\begin{algorithm}[!ht]
\caption{ReportMinCut($e,e'$)}
\tcp{Reports size of $(s,t)$-min-cut after failure of non-critical edges $e,e'$}
\vspace{2mm}
\If{\textup{Either $e$ or $e'$ is a critical edge}}
  {Return ``\em Invalid input''\;}
\If{\textup{Both $e$ and $e'$ are contained in $\bigcup_{f\in \A} \Null_G(f,\min+1)$}}
{
\If{\textup{$e'\notin \Null_G(f_e,\min+1)$ and endpoints of $e'$ are in different SCC in $G_{f_e}-e$}}
{Return $(\lambda-1)$\;}
}
{Return $(\lambda)$\;}
\label{alg:ReportMinCut-I}
\end{algorithm}

\subsection{Handling failing sets containing at least one critical edge}
\label{section:min-cut-1-or-more-critical-edges}

To handle the failure of one or more critical edges, we introduce a general data structure that, given any set of $k$ edge failures, determines whether the $(s, t)$-min-cut size in $G$ decreases by exactly $k$ in $O(k^2)$ time. This is formalized as follows (for proof see appendix).

\begin{theorem} 
For any $n$-vertex directed graph $G$ with source $s$, sink $t$, and $(s,t)$-min-cut of size $\lambda$, there exists an $O(\lambda n)$-sized oracle, ${\cal O}_{MINCUT}(s,t,G)$, that, for any set $F$ of $k$ edges, determines in $O(k^2)$ time whether the $(s,t)$-min-cut size decreases by $k$ upon the failure of $F$.

Furthermore, if the min-cut size decreases by exactly $k$, the oracle can compute and report an $(s,t)$-min-cut in $G - F$ in $O(kn)$ time.
\label{theorem:k-failures-k-diff} 
\end{theorem}

Now consider a failing set $F = \{e, e'\}$ comprising two edges in $G$, where at least one edge in~$F$, say $e$, is a critical edge. Upon the failure of $F$, the size of the $(s,t)$-min-cut decreases by at most two and at least one (since the min-cut size in $G - e$ is exactly $\lambda - 1$). Furthermore, using \Cref{theorem:k-failures-k-diff}, we can construct a data structure of size $O(\lambda n)$ that verifies, in constant time, whether the min-cut size decreases by exactly two when $F$ fails.

Combined with the discussion in \Cref{section:min-cut-no-critical-edges}, this gives an $O(n\lambda)$-sized data structure that, for any set $F$ of two edge failures, reports the size of the $(s,t)$-min-cut in $G - F$ in constant time.
The details of our data structure for reporting a min-cut after dual failures is deferred to the full version of the paper. This concludes the proof of \Cref{theorem:dual-min-cut}.

\section{Fault-tolerant Min-cut oracle for $k$ edge failures}
\label{section:FT-k-failures}

In this section, we present our $(s,t)$-min-cut sensitivity oracle resilient to $k$ failures. 
We begin with the following lemma (for proof see appendix).

\begin{lemma}
\label{lemma:k-failures-minimal-cut}
For any failing set $F$ of size $k$, $\mincut(s,t,G-F)$ is given by
$$
\min\Big\{ |C|-|F_0|~ \Big | ~C \text{ is minimal cut in $G$ of size at most }\lambda+k, F_0=F\cap C \Big\}.
$$
\end{lemma}

The above lemma highlights the significance of $\lambda+i$ minimal-cuts for answering min-cut queries under failures. 
In order to use the this lemma, one would require a data structure that can check, for every subset $F_0 \subseteq F$, whether there exists an $(s,t)$-minimal cut $C$ of size at most $\lambda + k$ containing $F_0$. While we provide a data structure in \Cref{theorem:k-failures-k-diff} to efficiently compute an $(s,t)$-min-cut containing a given subset $F_0$, extending this to handle all $\lambda + k$ minimal cuts is non-trivial.

To address this, we leverage the result by Kim et al.~\cite{KKPW-stoc22}, which shows how to augment graphs with additional edges to transform a large collection of minimal cuts into min-cuts. This allows us to efficiently utilize data-structure developed for min-cuts in \Cref{theorem:k-failures-k-diff}.

For any set $\E \subseteq V \times V$, let $G + \E^\infty$ denote the graph obtained by adding edges in $\E$ to $G$, 
with infinite capacity on all edges in $\E$.
We use the following result by Kim et al.~\cite{KKPW-stoc22}.

\begin{theorem}[\hspace{-0.1mm}\cite{KKPW-stoc22}]
There exists a randomized polynomial-time algorithm that, given a directed graph $G=(V,E)$, two vertices $s, t \in V$, and an integer $L$, outputs a set $\E \subseteq V \times V$ such that for every $(s,t)$-minimal-cut $Z \subseteq E$ of size at most $L$, with probability $2^{-O(L^4 \log L)}$, $Z$ remains an $(s,t)$-cut in $G + \E^\infty$, and furthermore, $Z$ is an $(s,t)$-min-cut in $G + \E^\infty$.
\label{theorem:flow-augmentation}
\end{theorem}

\paragraph{Oracle construction}
Let $L = \lambda + k$ and $\rho = 2^{O(L^4 \log L)} \cdot (4L \log_e n)$. We perform $\rho$ independent rounds, and in each round, we compute a sample $\E$ by invoking \Cref{theorem:flow-augmentation}. Let $\C$ be a collection of those samples $\E \in V \times V$ for which $G + \E^\infty$ has an $(s,t)$-min-cut of size at most $L$. For each $\E \in \C$, compute the oracle ${\cal O}_{MINCUT}(s,t,G + \E^\infty)$ using the data structure from \Cref{theorem:k-failures-k-diff}. Finally, verify in $m^{O(L)}$ time that for every $(s,t)$-minimal-cut $Z$ of size at most $L$, there exists some $\E \in \C$ satisfying $Z$ is an $(s,t)$-min-cut in $G + \E^\infty$. If not, re-compute~$\C$. The total space required by our data structure is $O(\rho\cdot nL)=\big(2^{O(L^4 \log L)}n \log n\big)$, since the graphs in family $\C$ have a min-cut of size at most $L$. \Cref{lemma:high-prob-bound} below shows that the number of repetitions needed to compute $\C$ is at most one with high probability.

\begin{lemma}
With probability at least $1 - \frac{1}{n^2}$, for each $(s,t)$-minimal cut $Z$ of size at most $L$, there exists $\E \in \C$ such that $Z$ is an $(s,t)$-min-cut in $G + \E^\infty$.
\label{lemma:high-prob-bound}
\end{lemma}

\begin{proof}
Consider an $(s,t)$-minimal cut $Z$. 
The probability that there does not exist an $\E \in \C$ such that $Z$ is an $(s,t)$-min-cut in $G + \E^\infty$ is

$$
\Big(1-\frac{1}{2^{O(L^4 \log L)}}\Big)^\rho 
= \Big(1-\frac{1}{2^{O(L^4 \log L)}}\Big)^{2^{O(L^4 \log L)} \cdot (4L \log_e n)}
\leq e^{-4L \log_e n}
= \frac{1}{n^{4L}}.
$$

The total number of minimal cuts of size at most $L$ is bounded by $n^{2L}$. Therefore, the probability that there exists a minimal cut $Z$ for which no $\E \in \C$ satisfies that $Z$ is an $(s,t)$-min-cut in $G + \E^\infty$ is at most $n^{-2L} \leq n^{-2}$.
This completes the proof.
\end{proof}

\paragraph{Query procedure}
Our algorithm to determine the size of the minimum cut after $k$ failures is presented in \Cref{alg:ReportMinCut-II}. The procedure iterates over every subset $F_0 \subseteq F$ and every set $\E \in \C$, and checks using \Cref{theorem:k-failures-k-diff} whether the size of the min-cut in $G + \E^\infty$ decreases by exactly $|F_0|$ upon the failure of $F_0$. 

It then reports the minimum value of  $\mincut(s,t,G+\E^\infty) - |F_0|$,  taken over all subsets $F_0 \subseteq F$ and all $\E \in \C$ such that the $(s,t)$-min-cut size in $G + \E^\infty$ decreases by exactly $|F_0|$ upon the failure of $F_0$.
The time complexity to compute the size of the min-cut in $G - F$ is therefore $O(k^2 |\C|) = O\big( 2^{O(L^4 \log L)} \log n \big)$, where $L = \lambda + k$.

Furthermore, we can compute an explicit $(s,t)$-min-cut in $G - F$. To achieve this, we return a min-cut in the graph $G + \E^\infty - F_0$ for the appropriate choice of $\E$ and $F_0$, as determined in the previous step. Computing and returning this cut requires an additional $O(kn)$ time using the oracle described in Theorem~\ref{theorem:k-failures-k-diff}.

\begin{algorithm}[!ht]
\setstretch{1.25}
\caption{ReportMinCut($F$)}
Initialize $q=\lambda$\;
\ForEach{$F_0\subseteq F$ {\bf and} $\E\in \C$}{
\If{$(s,t)$-min-cut size in $G+(\E)^\infty$ on failure of $F_0$ decreases by exactly $|F_0|$}
{$q=\min\Big(q,~\mincut(s,t,G+\E^\infty)-|F_0|\Big)$}
}
{Return $q$\;}
\label{alg:ReportMinCut-II}
\end{algorithm}

We thus conclude with the following theorem.

\begin{theorem}
There exists a Las Vegas algorithm that, for any $n$-vertex directed graph $G$ with an $(s,t)$-min-cut of size $\lambda$, computes a $k$-fault-tolerant minimum cut oracle of $O\big( 2^{O(L^4 \log L)} n \log n \big)$ space. This oracle can determine the size of the $(s,t)$-min-cut in $G - F$ for any set $F$ of $k$ edges in $O\big( 2^{O(L^4 \log L)} \log n \big)$ time, where $L = \lambda + k$.

Furthermore, the oracle can report an $(s,t)$-min-cut in $G - F$ in an additional computation time of $O(kn)$.
\label{theorem:min-cut-k-failures}
\end{theorem}

\section{Open Problems}
\label{section:open-problems}
Several interesting questions remain unresolved in this field. One immediate open question after our work is whether our max-flow sensitivity oracle can be extended to handle more than two failures. Specifically, can we obtain a compact oracle capable of updating the max-flow in a directed graph after $k > 2$ failures in $O_k(n)$ time?

Another natural open question is computing a compact min-cut sensitivity oracle for graphs with large cuts for $k>2$ failures. While our min-cut sensitivity oracle for $k = 2$ has a linear dependence on $\lambda$, our oracle for $k>2$ exhibits an exponential dependence on~$\lambda$. We leave it open to close this gap and achieve a polynomial dependence on $\lambda$.

A further open question is whether our $(s,t)$-reachability oracle (as well as the max-flow and min-cut sensitivity oracles) can be extended to the single-source setting. Currently, for $k = 1, 2$, there exist single-source reachability oracles~\cite{LengauerT:79, Choudhary16} with linear space that for any query vertex $x$ and any set $F$ of up to two failures answer whether $x$ is reachable from source on failure of $F$ in constant time. However, to date, no oracle exists for answering reachability queries from a fixed source in linear space and truly sub-linear query time, even for a constant $k > 2$.

\bibliographystyle{plain}
\bibliography{ref}

\newpage
\appendix

\section{Omitted Proofs}
\label{section:deferred-proofs}

\subsection{Proof of Lemma \ref{lemma:null-set-max-flow}}

Consider an $(s,t)$-max-flow $f$ and an edge $e = (a,b) \in{\Null}_G(f, \min+1)$. Clearly, $e$ must be an intra-cluster edge. Let $C$ be an $(s,t)$-minimal cut of size $(\lambda + 1)$ containing $e$. Since $e$ is a non-critical edge in a $\lambda+1$-cut, removing $e$ from $C$ results in an $(s,t)$-min-cut that separates the endpoints of $e$. Therefore, the endpoints of $e$ lie in different SCCs in $(G - e)_f = G_f - e$. This shows that the deletion of any $e \in {\Null}_G(f, \min+1)$ from $G$ splits an SCC of $G_f$.

Next let $H_f$ be an SCC-certificate of $G_f$ with at most $2n$ edges satisfying that any two vertices that are strongly connected in $G_f$ are also strongly connected in $H_f$. Such a certificate $H_f$ with at most $2n$ edges can be constructed by simply taking each SCC $S$ in $G_f$, and including in $H_f$, the edges of an in-reachability and out-reachability tree of the induced graph $G_f[S]$ rooted at an arbitrary vertex in $S$.

As deletion of any edge in $\Null_{G}(f,\min+1)$ splits an SCC of $G_f$, the edges in $\Null_{G}(f,\min+1)$ must be contained in every SCC certificate of $G_f$. This proves that $|\Null_{G}(f,\min+1)|\leq |E(H_f)|\leq 2n$. 
\vspace{1mm}

\begin{reminder}{Lemma \ref{lemma:null-set-max-flow}}
For any integer $(s,t)$-max-flow $f$ in an unweighted graph $G$, the cardinality of the set $\Null_{G}(f,\min+1)$ is at most $2n$.
\end{reminder}

\subsection{Proof of Theorem \ref{theorem:k-failures-k-diff}}

Consider a poset relation on the edges of $D_\lambda$ such that, for any two edges $e_x$ and $e_y$ in $D_\lambda$, $e_x \leq e_y$ if and only if there exists a path from $\bm{s}$ to $\bm{t}$ in $D_\lambda$ in which $e_x$ precedes $e_y$.
This relation induces a partial order because $D_\lambda$ is acyclic. Recall that a {\em chain} in a poset is a subset in which every two elements are comparable, while an {\em anti-chain} is a subset in which no two elements are comparable. Due to \Cref{property:transversal}, the following is immediate.

\begin{lemma}
A set $C$ of edges in $G$ is an $(s,t)$ minimum-cut if and only if $C$ corresponds to a maximal anti-chain of critical edges in~$D_\lambda$.
\label{lemma:min-cut-characterization}
\end{lemma}

\begin{lemma}
Let $\E$ be a set of $k$ edges in $G$. Then on deletion of $\E$, the size of min-cut reduces~by~$k$ if and only if all the edges in $\E$ are critical edges forming an anti-chain in $D_\lambda$.
\label{lemma:k-failures-property}
\end{lemma}

\begin{proof}
Consider a set $\E$ of $k$ edges. Let us first consider the case that all the edges in $\E$ are critical and form an anti-chain in $D_\lambda$. As any anti-chain can be greedily extended to a maximal anti-chain, let $\E_0$ be a maximal anti-chain comprising of critical edges in $D_\lambda$ satisfying $\E_0\supseteq \E$. Due to \Cref{lemma:min-cut-characterization}, $\E_0$ is an $(s,t)$ min-cut in $G$. As all the edges in $\E$ are contained in an $(s,t)$-min-cut, the size of $(s,t)$-min-cut reduces by exactly $k$ on failure~of~$\E$.

Next, suppose the size of min-cut reduces by exactly $k$ on failure of $\E$. Let $\E_1$ be an $(s,t)$-cut in $G- \E$ of size $\lambda-k$. Then, $\E\cup \E_1$ will be an $(s,t)$-cut of size $\lambda$. Therefore, by \Cref{lemma:min-cut-characterization}, the edges of $\E\cup \E_1$ (and hence also $\E$) are critical and form an anti-chain in $D_\lambda$.
\end{proof}

In order to exploit \Cref{lemma:k-failures-property}, we need an efficient data structure to efficiently verify chain/anti-chain relation among critical edges in $D_\lambda$. We present in the following lemma an $O(\lambda n)$ sized data-structure that achieves this task in constant time.

\begin{lemma}
The graph $D_\lambda$ can be processed in polynomial time to compute a data-structure of size $O(\lambda n)$ that given any two critical edges $e_a,e_b$ determines if there is a path in $D_\lambda$ in which $e_a$ precedes $e_b$ in $O(1)$ time.
\label{lemma:DAG-reachability}
\end{lemma}
\begin{proof}
Consider a family $\P=(P_1,\ldots,P_{\lambda})$ of $\lambda$ edge-disjoint paths from source to destination in $D_\lambda$. For~each equivalence class $\bm x\in \W$ and each path $P\in \P$, let $F(\bm x, P)$ denote the first equivalence class in path $P$ that is reachable from $\bm x$ in graph $D_\lambda$. Consider two critical edges $e_a=(\bm a,\bm a')$ and $e_b=(\bm b,\bm b')$ in $D_\lambda$, and let $P$ be the path in $\P$ containing the edge $e_b$. To determine whether $e_a\leq e_b$ it suffices to verify if $F(\bm a', P)$ is either identical to or a predecessor of $\bm b$ in $P$. To~verify this condition in constant time, we associate with each equivalence class $\bm w$ and each path $P\in \P$, the rank of $\bm w$ in $P$. Thus, we can easily compare the ranks of $F(\bm a', P)$ and $\bm b$ in $P$ to determine the relationship between $e_a$ and $e_b$ in constant time. The~size of the data structure is $O(\sum_{i\leq \lambda}|V(P_i)|)$, plus the space needed to store for each $\bm x\in \W$ and each $P\in \P$, the node $F(\bm x,P)$ and rank of $\bm x$ in $P$. This takes a total of $O(\lambda n)$ space. This completes the proof of the claim.
\end{proof}

On combining \Cref{lemma:k-failures-property} and \Cref{lemma:DAG-reachability}, we obtain an $O(\lambda n)$-sized oracle that, for any set $F$ of $k$ edges, determines in $O(k^2)$ time whether the $(s,t)$-min-cut decreases by $k$ upon the failure of $F$.

It remains to show how to compute an $(s,t)$-min-cut in $G - F$ when the min-cut size decreases by exactly $k$ due to the failure of a set $F$ of $k$ edges. Observe that any min-cut in $G$ that includes all the edges of $F$ must also be an $(s,t)$-min-cut in $G - F$. Thus, it suffices to output a min-cut in $G$ that contains the edges of $F$.

For any vertex $x \in V$, let $\R_{in}(x)$ denote the set of vertices $u \in V$ such that there exists a path from $\bm u$ to $\bm x$ in the strip graph $D_\lambda$. The following lemma provides a characterization of the nearest min-cut when a set of vertices is merged with the source vertex $s$.

\begin{lemma}[Dinitz and Vainshtein~\cite{DinitzV00}]
\label{lemma:DV-NMC-FMC}
For any distinct vertices $x_1,\ldots,x_k$ in $G$ not lying in the equivalence class $\bm t$, the union $$\bigcup_{i=1}^k \R_{in}(x_i)$$
corresponds to the source side of $\nmc\Big(\{s,x_1,\ldots,x_k\},t,G\Big)$. 
\end{lemma}

To compute the source side of $\nmc(s,t,G-F)$, we iterate over all vertices $v \in V$ and check, using \Cref{lemma:DAG-reachability}, whether $\bm v$ is identical to or precedes $\bm{x_i}$, for some $i \in [1,k]$. This check takes $O(k)$ time per vertex. If the condition is satisfied, we add $v$ to the source side $A$. Finally, we return the cut $(A,A^c)$. The total time complexity for this computation is $O(kn)$.

We conclude with the following result.

\begin{reminder}{Theorem~\ref{theorem:k-failures-k-diff}} 
For any $n$-vertex directed graph $G$ with source $s$, sink $t$, and $(s,t)$-min-cut of size $\lambda$, there exists an $O(\lambda n)$-sized oracle, ${\cal O}_{MINCUT}(s,t,G)$, that, for any set $F$ of $k$ edges, determines in $O(k^2)$ time whether the $(s,t)$-min-cut size decreases by $k$ upon the failure of $F$.

Furthermore, if the min-cut size decreases by exactly $k$, the oracle can compute and report an $(s,t)$-min-cut in $G - F$ in $O(kn)$ time.
\end{reminder}

\subsection{Proof of \Cref{lemma:k-failures-minimal-cut}}

Let $F_0$ be a subset of $F$ of minimum size such that the $(s,t)$-min-cut size in $G - F_0$ is identical to that in $G - F$.
Let $Z$ denote an $(s,t)$-min-cut in $G - F_0$ (and hence also in $G - F$). Then, $Z \cup F_0$ forms an $(s,t)$-cut in $G$ with size at most $\lambda + k$. Moreover, $Z \cup F_0$ is a minimal cut; otherwise, there would exist an edge $e \in Z \cup F_0$ such that either $Z\setminus\{e\}$ is a cut in $G - F_0$, or the $(s,t)$-min-cut values in $G - F_0$ and $G - (F_0 \setminus\{e\})$ would coincide.

Thus, by considering all minimal $(s,t)$-cuts of size at most $\lambda + k$ in $G$ and removing $F$ from each, one of them must coincide with $Z$. Furthermore, no smaller cut than $Z$ can be generated, as this would contradict the assumption that $Z$ is an $(s,t)$-min-cut in $G - F$.

\begin{reminder}{Lemma~\ref{lemma:k-failures-minimal-cut}} 
For any failing set $F$ of size $k$, $\mincut(s,t,G-F)$ is given by
$$
\min\Big\{ |C|-|F_0|~ \Big | ~C \text{ is minimal cut in $G$ of size at most }\lambda+k, F_0=F\cap C \Big\}.
$$
\end{reminder}

\section{Lower Bound Results}
\label{section:lower_bounds}

\subsection{Dual Fault-Tolerant Min-cut Oracle}
In this subsection, we show that the construction of our dual fault-tolerant min-cut oracle of $O(\lambda n)$ size is essentially tight up to the word size.

\begin{figure}[!ht]
\centering
\includegraphics[scale=.4]{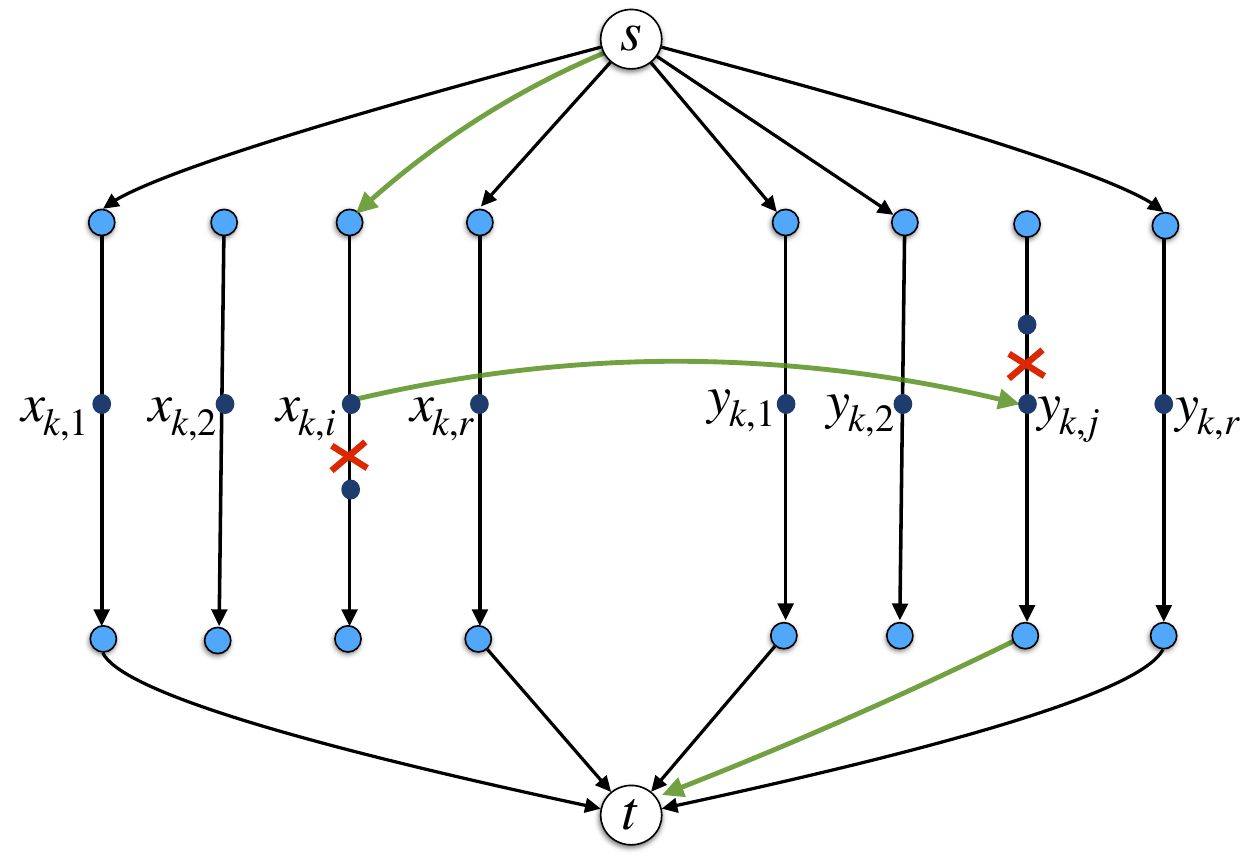}
\caption{Lower bound construction for dual fault-tolerant min-cut oracle.}
\label{fig:lower-bound}
\end{figure}

Let $r,L,n$ be positive integers satisfying $n=\Theta(Lr)$, and $Z_1,\ldots,Z_L$ be arbitrary $r\times r$ binary matrices.
We construct a directed graph $G = (V, E)$ with $n$ vertices, source $s$, and sink $t$ as follows.
We take $2r$ internally vertex-disjoint directed paths: $X_1, \ldots, X_r$ and $Y_1, \ldots, Y_r$, each of length $L + 1$, where, $X_i = (s, x_{1,i}, x_{2,i}, \ldots, x_{L,i}, t)$ and $Y_j = (s, y_{1,j}, y_{2,j}, \ldots, y_{L,j}, t)$, for $i, j \in [r]$. The vertex set $V(G)$ is the union of all vertices on the paths $X_i$ and $Y_j$.
The edge set $E(G)$ includes all edges from the $2r$ paths, as well as edges defined by the matrices $Z_k[i,j]$. Specifically, for each $i, j \in [r]$ and each $k \in [L]$, we add an edge $(x_{k,i}, y_{k,j})$ to $G$ if and only if the $(i,j)^{\text{th}}$ bit of matrix $Z_k$ is 1. This construction is illustrated in \Cref{fig:lower-bound}.

We claim that the $(i,j)^{\text{th}}$ bit of any matrix $Z_k$ can be computed using dual fault-tolerant min-cut queries.
Specifically, on failure of edges $e_1=(x_{k,i}, x_{k+1,i})$ and $e_2=(y_{k-1,j}, y_{k,j})$, the $(s,t)$-max-flow decreases by one if edge $(x_{k,i}, y_{k,j})$ exists; otherwise, it decreases by two. 
This is because if $(x_{k,i},y_{k,j})\in E(G)$, then upon the failure of $e_1$ and $e_2$, the flow entering node $x_{k,i}$ can be rerouted through edge $(x_{k,i},y_{k,j})$ to exit from node $y_{k,j}$, reducing the $(s,t)$-max-flow in $G$ by exactly one. Such rerouting is not feasible if edge $(x_{k,i},y_{k,j})$ is not in $G$. 

Therefore, given access to a dual fault-tolerant min-cut oracle, we can determine all $r \times r$ entries of the matrices $Z_1, \ldots, Z_L$. We summarize this result in the following theorem.

\begin{theorem}
\label{theorem:lb-oracle-even-flow}
For any positive integers $n$ and $r$ $(r \leq n)$, there exists an $n$-vertex directed graph with a source $s$, a sink $t$, and an $(s,t)$-max-flow of $2r$, such that any dual fault-tolerant min-cut oracle for this graph requires $\Omega(nr)$ bits.
\end{theorem}

This result holds even if the $(s,t)$-max-flow is increased to $2r + 1$ by adding an edge directly from the source to the sink. This proves \Cref{theorem:lb-min-cut}.

\subsection{A Lower Bound on Families $\A$ and $\B$ under Dual Failures}

In this subsection, we establish a lower bound of $\Omega(n)$ on the size of the families $\A$ and $\B$ in the presence of dual edge failures.

Consider a directed graph $G = (V, E)$ with a designated source $s$ and sink $t$. Throughout this section, we use $\A$ to denote the family of maximum flows in $G$ such that for any set $F\subseteq E$ of size two satisfying $\maxflow(s,t,G)=\maxflow(s,t,G-F)$, the family $\A$ contains a maximum flow of $G - F$. Additionally, we use $\B$ to represent the family of maximum flows in $G - F$, for each set $F\subseteq E$ of size two. We aim to show that both $\A$ and $\B$ must have a size of at least $\Omega(n)$.

We construct a directed graph $G = (V, E)$ with $n$ vertices as follows: The graph contains two internally vertex-disjoint directed paths, $X = (x_1, x_2, \ldots, x_L)$ and $Y = (y_1, y_2, \ldots, y_L)$, each of length $L = (n - 2)/2$. These paths share their endpoints, that is, $x_1 = y_1$ and $x_L = y_L$. Additionally, an edge is added from the source $s$ to $x_1$, from $x_L$ to the sink $t$, and from $x_i$ to $y_i$, for each $i \in [2, L - 1]$.

It is straightforward to verify that the minimum $(s, t)$-cut size in this graph is one.

Now consider the failure of two edges: $(x_i, x_{i+1})$ and $(y_{i-1}, y_i)$. Under these failures, the path 
$ R_i := (s, x_1, x_2, \ldots, x_i, y_i, y_{i+1}, \ldots, y_L, t)$ becomes the unique $(s, t)$-path in the graph. This implies that there are at least $L = \Omega(n)$ distinct maximum flows required in family $\A$ as well as in family $\B$.

We thus obtain the following theorem.

\begin{theorem}
For any positive integer $n$, there exists an $n$-vertex directed graph with a source~$s$, a sink $t$, and an $(s,t)$-maximum flow of value one such that any construction of the families $\A$ and $\B$ resilient to dual failures require at least $\Omega(n)$ flows.
\label{theorem:FT-family-lower-bound}
\end{theorem}

\section{Circulation with Lower Limits}
\label{section:circulation-with-lower-limits}

The circulation problem with lower bounds is a specific type of network flow problem. Consider a directed graph $G = (V, E,d,\ell,\mu)$ such that $d$ is a function mapping $V$ to real numbers, and $\ell,\mu$ are both functions mapping $E$ to non-negative reals. The objective here is to find a flow $f: E \to \mathbb{R}^+\cup\{0\}$ that satisfies:
\subparagraph{1. Capacity constraints} $ \ell(e) \leq f(e) \leq \mu(e) $, for each $ e\in E $.
\subparagraph{2. Flow conservation} For each $ v \in V $,
$$
\sum_{x \in \inv(v)} f(x,v) - \sum_{y \in \outv(v)} f(v,y) ~=~ d(v).
$$

The criteria for a circulation to exist is fairly well-known; however, for completeness, we will prove the necessary and sufficient conditions in the following lemma.

\begin{lemma}
Let $G=(V,E,d,\ell,\mu)$ be an instance of the circulation problem. A circulation for $G$ exists if and only if $\sum_{v\in V}d(v)=0$, and, for all cuts $(A,B)$ in $V$, the following inequality holds.
$$
d(B)+\ell(B,A)~\leq~ \mu(A,B),
$$
where, 
$$
d(B)=\sum_{v\in B} d(v)
\text{,~~~~} 
\ell(B,A)=\displaystyle \sum_{\substack{(x,y)\in E,\\x\in B,~y\in A}}\ell(x,y)
\text{,~~~and~~~} 
\mu(A,B)=\displaystyle \sum_{\substack{(x,y)\in E,\\x\in A,~y\in B}}\mu(x,y).
$$
\label{lemma:circulation-with-lower-limits}
\end{lemma}
\begin{proof}
Consider an instance $G=(V,E,d,\ell,\mu)$ of the circulation problem. We first define a new instance $G^*=(V,E,\delta,c)$ of the circulation problem (without lower limits)  as follows. 
For each $v\in V$, let
$$\delta(v)~=~d(v)~+~\sum_{y\in \outv(v)}\ell(v,y)~-~\sum_{e\in \inv(v)}\ell(x,v).$$

The rationale behind defining $\delta(v)$ is that if, for each $e\in E$, we pass $\ell(e)$ amount of flow, then the problem reduces to a circulation problem without lower limits where, for any node~$v$, the demand of $v$ increases by the value
$~\sum_{e\in \outedges(v)}\ell(e)~-~\sum_{e\in \inedges(v)}\ell(e)$. 
Note that in doing so, the capacity of any edge $e$ must reduce by $\ell(e)$. Thus, for each $e\in E$, we set $c(e)=\mu(e)-\ell(e)$. Observe that a circulation in $G$ exists if and only if a circulation in $G^*$ exists.\\

Now, to solve the circulation in $G^*$, compute a graph $\widetilde G=(\widetilde V,\widetilde E,c)$ as follows. The set~$\widetilde V$ is obtained by adding two additional nodes $s,t$ to $V$. We include in $\widetilde E$ the edges in the set $E$. In addition,
\begin{itemize}
\item for each $x\in V$ with $\delta(x)<0$, add an edge from $s$ to $x$ and set $c(s,x)=|\delta(x)|$;
\item for each $y\in V$ with $\delta(y)>0$, add an edge from $y$ to $t$ and set $c(y,t)=|\delta(y)|$.
\end{itemize}

\vspace{2mm}
Let $V^-$ (resp. $V^+$) be the sets of nodes $v$ for which $\delta(v)$ is negative (resp. positive). It is easy to observe that the new circulation problem (without lower limits) is solvable if and only if the following holds.

\begin{enumerate}
\item $\sum_{v\in V}\delta(v)=0$, or equivalently, $\sum_{v\in V}d(v)=0$; and
\item there is an $(s,t)$ flow in $\widetilde G$ of value at least $\sum_{v\in V^+}\delta(v)$.\\[-2mm]
\end{enumerate}

\noindent
The second condition can be rewritten as follows: for all partitions $(A,B)$ of $V$,
the capacity of  cut $(s\cup A,B\cup t)$ is at least $\sum_{v\in V^+}\delta(v)$. That is,

\begin{align*}
\sum_{v\in V^+}\delta(v)~\leq~~ c(s\cup A,B\cup t)~~
&=~~ \sum_{v\in V^- \cap B}c(s,v)+c(A,B)+\sum_{v\in V^+ \cap A}c(v,t)\\[2mm]
&=~~ -\sum_{v\in V^- \cap B}\delta(v)+c(A,B)+\sum_{v\in V^+ \cap A}\delta(v).
\end{align*}

\noindent
This implies: 
\begin{align*}
c(A,B)~~&\geq~~\sum_{v\in V^+}\delta(v)-\sum_{v\in V^+ \cap A}\delta(v) +\sum_{v\in V^- \cap B}\delta(v)\\
&=~~ \sum_{v\in B}\delta(v)\\
&=~~ \sum_{v\in B} \Big(d(v)~+~\sum_{y\in \outv(v)}\ell(v,y)~-~\sum_{e\in \inv(v)}\ell(x,v)\Big)\\
&=~~ d(B)~+~~\ell(B,A)~~-~~\ell(A,B).
\end{align*}

\vspace{2mm}

\noindent
Therefore, a circulation for $G=(V,E,d,\ell,\mu)$ exists if and only if $\sum_{v\in V}d(v)=0$, and $\mu(A,B)\geq d(B)+\ell(B,A)$, for every cut $(A,B)$ in $G$.
\end{proof}

\section{Review of $\Omega(n^2)$ space bound in Min-cut Sensitivity Oracle of \cite{BaswanaBP22}}
\label{section:baswana-min-cut-oracle-space}
 
A natural question that arises is what happens to the space complexity of the  dual fault-tolerant oracle of Baswana et al.~\cite{BaswanaBP22} on parameterizing the problem by the $(s,t)$-min-cut size $\lambda$? Specifically, if $\lambda = o(n)$, does the oracle of \cite{BaswanaBP22} still require $\Omega(n^2)$ space?

We demonstrate that this is indeed the case. To establish this, we construct a graph $G$ with an $(s,t)$-min-cut of size 1, for which the oracle in \cite{BaswanaBP22} requires $\Theta(n^2)$ space, even when both deleted edges are intra-cluster and belong to the same equivalence class. In contrast, since $\lambda = 1$, our oracle, as presented in \Cref{theorem:dual-min-cut}, requires only $O(n)$ space.

To provide context, we briefly outline what the oracle in \cite{BaswanaBP22} stores for handling {\em intra-cluster} edge failures. Since our focus is on analysing the space complexity of oracle, we omit details of their query procedure.

For each cluster $W \in \W$, let $u_W \in W$ be an arbitrary representative vertex. To handle queries where both failed edges are intra-cluster and lie in the same cluster $W$, the oracle in \cite{BaswanaBP22} stores the following:

\begin{itemize}
\item For each $x\in W$, let ${\cal C}^N_{x, u_W}$ be the family of nearest $(\min+1)$-cuts in $G$ such that $x, u_W$ lie on the source side. Then, the oracle stores the family 
${\cal N}_{x, u_W} = \{ B_C\cap W ~|~ (A_C,B_C) \in {\cal C}^N_{x, u_W} \}$.
\item For each $x\in W$, let ${\cal C}^F_{x, u_W}$ be the family of farthest $(\min+1)$-cuts in $G$ such that $x, u_W$ lie on the sink side. Then, the oracle stores the family 
${\cal F}_{x, u_W} = \{ A_C\cap W ~|~ (A_C,B_C) \in {\cal C}^F_{x, u_W} \}$.
\end{itemize}

Baswana et al.~\cite{BaswanaBP22} proved that for a given $x \in W$, all sets in the family ${\cal N}_{x, u_W}$ are disjoint, allowing it to be stored in $O(|W|)$ space. Similarly, storing ${\cal F}_{y, u_W}$ for each $y \in W$ also requires $O(|W|)$ space. Thus, for handling failures within induced graph $G[W]$, their oracle requires a total of $O(|W|^2)$ space. We show that this bound is tight even when $\lambda = 1$.

\begin{figure}[!ht]
\centering
\includegraphics[scale=.38]{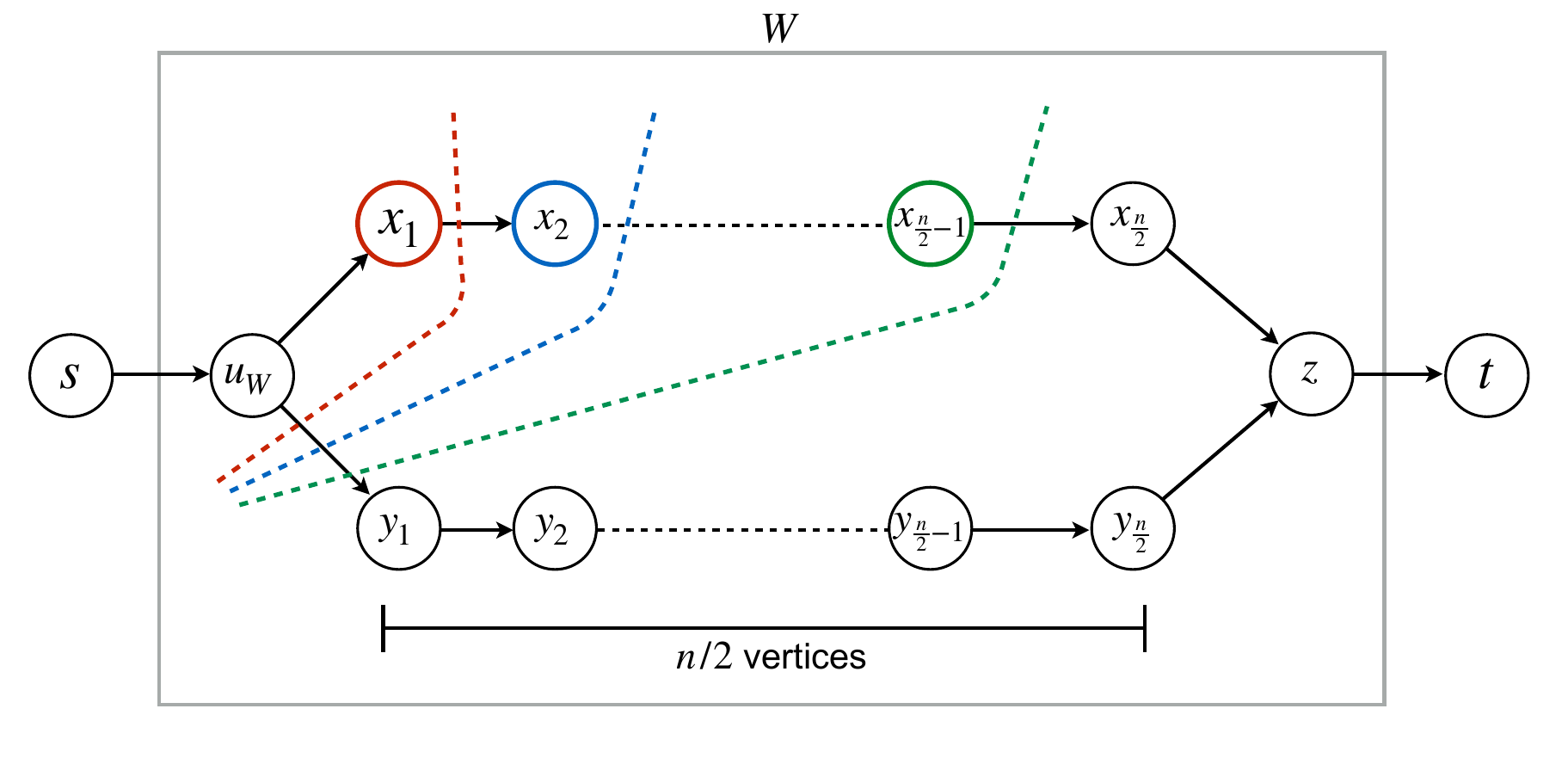}
\caption{Depiction of nearest $(\min+1)$-cuts in $G$ such that $x_i,u_W$ lie on the source side.}
\label{fig:baswana-nearest-cuts}
\end{figure}
\vspace{-3mm}

\paragraph{Tight Space Bound Example}
Consider a graph $G$ with $n+4$ vertices with a source $s$, sink~$t$, and two internally vertex-disjoint paths—one from $u_W$ to $z$ via vertices $x_1, x_2, ..., x_{n/2}$, and another from $u_W$ to $z$ via vertices $y_1, y_2, ..., y_{n/2}$. Additionally, there is an edge from $s$ to $u_W$ and from $z$ to $t$. See \Cref{fig:baswana-nearest-cuts}. The graph has an $(s,t)$-min-cut size of 1.

Now consider the unique equivalence class $W$ in $G$ that does not contain $s$ and $t$. For each vertex $x_i \in W$, there exists a unique nearest $(\min+1)$-cut such that both $x_i$ and $u_W$ lie on the source side. For each $i$, the family 
$ {\cal N}_{x_i,u_W}$ contains the set $\{ x_{i+1}, x_{i+2}, \dots, x_{n/2}, y_1, y_2, \dots, y_{n/2}, z \}$ of size at least $n/2$. Therefore, storing all such sets for each vertex in $W$ requires at least $\Omega(n^2)$ space.

This shows that even when $\lambda = 1$, the dual fault-tolerant min-cut oracle of~\cite{BaswanaBP22} requires quadratic space.

\end{document}